\documentclass[11pt]{article}
\usepackage{fullpage}
\usepackage[margin=1in]{geometry}
\usepackage[numbers]{natbib}
\usepackage[utf8]{inputenc}
\usepackage{amsmath,amsthm,amssymb}
\usepackage{thmtools} 
\usepackage{thm-restate}
\usepackage{amsfonts,bm,bbold,dsfont}
\usepackage{enumitem}
\usepackage[normalem]{ulem}
\usepackage{xspace}
\usepackage[usenames,dvipsnames]{xcolor}
\usepackage[colorlinks=true,citecolor=ForestGreen,linkcolor=ForestGreen,urlcolor=NavyBlue]{hyperref}

\usepackage[ruled]{algorithm2e} 
\SetAlFnt{\small}
\SetAlCapFnt{\small}
\SetAlCapNameFnt{\small}
\SetAlCapHSkip{0pt}
\IncMargin{-\parindent}

\usepackage[nameinlink,capitalise]{cleveref}

\usepackage{natbib}

\usepackage{comment} 
\usepackage{subcaption}
\usepackage{multirow}
\usepackage{mathrsfs}
\usepackage{bbold}
\usepackage{pgf,tikz,pgfplots}
\usetikzlibrary{matrix}
\usetikzlibrary{calc,patterns,intersections,pgfplots.fillbetween}
\pgfplotsset{compat=1.18}

\usepackage{times}
\usepackage{soul}
\usepackage{url}
\usepackage[utf8]{inputenc}
\usepackage{caption}
\usepackage{graphicx}
\usepackage{subcaption}
\usepackage{booktabs}
\usepackage{algorithmic}
\usepackage[switch]{lineno}
\usepackage{float}

\newcommand{\doctors}{\mathcal{D}}
\newcommand{\hospitals}{\mathcal{H}}
\newcommand{\hist}{\mathrm{Hist}}
\newcommand{\pull}{\mathrm{next}}
\newcommand{\accept}{\mathrm{accept}}
\newcommand{\rank}{\mathrm{rank}}

\newcommand{\E}{\mathop{\mathds{E}}}
\newtheorem{definition}{Definition}[section]
\newtheorem{lemma}{Lemma}[section]
\newtheorem{theorem}{Theorem}[section]
\newtheorem{claim}{Claim}[section]
\newtheorem{corollary}{Corollary}[section]

\newtheorem{observation}{Observation}[section]

\usepackage[]{color-edits}

\addauthor{AR}{blue}
\newcommand{\aarc}[1]{{\ARcomment{#1}}}
\newcommand{\aare}[1]{{\ARedit{#1}}}

\addauthor{SM}{green!50!black}

\addauthor{AE}{magenta}

\title{Stable Marriage: Loyalty vs. Competition
\thanks{The work of A.\ Eden and A.\ Ronen was supported by the Israel Science Foundation (grant No. 533/23).}}

\author{
		Amit Ronen
		\thanks{The Hebrew University; {\tt 3ortima@gmail.com}}
		\and
		Jonah Evan Hess
		\thanks{Unaffiliated; {\tt Jonahhessdev@gmail.com}}
            \and
		Yael Belfer
		\thanks{The Hebrew University; {\tt yael.belfer@mail.huji.ac.il}}
            \and
            Simon Mauras
            \thanks{INRIA (FairPlay); {\tt simon.mauras@inria.fr}}
            \and
		Alon Eden
		\thanks{The Hebrew University; {\tt alon.eden@mail.huji.ac.il}. Incumbent of the Harry \& Abe Sherman Senior Lectureship at the School of Computer Science
and Engineering at the Hebrew University.}
	}

\begin{document}

\begin{titlepage}

\maketitle

\begin{abstract}
We consider the stable matching problem (e.g. between doctors and hospitals) in a one-to-one matching setting, where preferences are drawn uniformly at random. It is known that when doctors propose and the number of doctors equals the number of hospitals, then the expected rank of doctors for their match is $\Theta(\log n)$, while the expected rank of the hospitals for their match is $\Theta(n/\log n)$, where $n$ is the size of each side of the market. However, when adding even a single doctor, [Ashlagi, Kanoria and Leshno, 2017] show that the tables have turned: doctors have expected rank of $\Theta(n/\log n)$ while hospitals have expected rank of $\Theta(\log n)$. That is, (slight) competition has a much more dramatically harmful effect than the benefit of being on the proposing side. Motivated by settings where agents inflate their value for an item if it is already allocated to them (termed endowment effect), we study the case where hospitals exhibit ``loyalty".

We model loyalty as a parameter $k$, where a hospital currently matched to their $\ell$th most preferred doctor accepts proposals from their $\ell-k-1$th most preferred doctors. Hospital loyalty should help doctors mitigate the harmful effect of competition, as many more outcomes are now stable. However, we show that the effect of competition is so dramatic that, even in settings with extremely high loyalty, in unbalanced markets, the expected rank of doctors already becomes $\tilde{\Theta}(\sqrt{n})$ for loyalty $k=n-\sqrt{n}\log n=n(1-o(1))$.
\end{abstract}
\end{titlepage}

\section{Introduction}

Two-sided markets capture many real-life scenarios where every entity on each side has a preference order on the other side. Some examples of such markets include doctor to hospital residency matching and student to school matching. A desired property of an allocation in such a market is \textit{stability}, which requires that no pair (e.g. a doctor and a hospital) would rather ``elope" and match each other rather than their prescribed match in the allocation. The famous Deferred Acceptance algorithm (henceforth referred to as DA), discovered by~\citet{gale1962college}, finds such a stable matching by letting one side, e.g. doctors, propose to the entities on the other side, e.g. hospitals. A hospital tentatively accepts the proposal whenever the doctor making the proposal is more attractive than their current match. A doctor, on the other hand, proposes to the most preferred hospital they haven't proposed to yet (more on this in Section~\ref{sec:prelims}).  It was established that the proposing side has an inherent advantage; namely,~\cite{gale1962college} show that when doctors propose, the resulting stable matching is the doctor-optimal matching in the set of all stable matchings. Moreover, the same matching is also the hospital-pessimal stable matching.

In order to quantify the benefit of being the proposing side,~\citet{wilson1972analysis,Pittel89} study one-to-one matching markets with $n$ doctors and $n$ hospitals and randomly drawn preferences. They reveal that when each doctor and hospital independently draws a permutation uniformly at random as preference order  hospitals and doctors, respectively, then the expected (average) rank of doctors for their match is $\Theta(\log n),$ while  the expected rank for hospitals is $\Theta(n/\log n)$. \citet{ashlagi2017unbalanced} study the same setting, but with just one additional doctor; that is, we have $n+1$ doctors and $n$ hospitals whose preferences are drawn independently and uniformly at random and doctors propose to hospitals. Surprisingly, they demonstrate that being on the slightly shorter side is much more beneficial than being the side making the offers. Namely, they show that in this case, the expected rank of doctors is $\Theta(n/\log n),$ while the expected rank of hospitals is $\Theta(\log n)$. 

The intuition behind this is the following. In a balanced market, the DA algorithm terminates when the last hospital gets a proposal, which can be proven to take $O(n\log n)$ proposals using a coupon-collector-style analysis. In contrast, with one extra doctor, once $n$ hospitals are matched, there is still an unmatched doctor who then tries to propose to all hospitals before the algorithm terminates. This doctor is very likely to get accepted by one of the hospitals, which results in another doctor getting rejected and trying to propose to all hospitals, and so on and so forth. ~\cite{ashlagi2017unbalanced} prove that this results in very long rejection chains, terminating when some doctor is rejected by all hospitals. This results in additional $\Omega(n^2/\log n)$ proposals. Motivated by behavioral studies indicating that individuals tend to inflate the value of items currently assigned to them, we investigate whether introducing the notion of \textit{loyalty} into this setting can mitigate the dramatic effects of \textit{competition}.

\paragraph{Loyalty vs. Competition.} Experimental studies consistently reveal that possessing items increases their perceived value~\cite{kahneman1991anomalies,knetsch1989endowment}, a phenomenon termed \textit{The Endowment Effect}~\cite{thaler1980toward}. Recent studies on market equilibrium suggest that even a mild form of  endowment effect might increase the set of stable outcomes beyond gross-substitutes valuations~\cite{babaioff2018combinatorial,EzraFF20}. We expect to find endowment effect in two-sided markets when DA is run in an online manner, such as in the French~\cite{parcoursup} or German~\cite{grenet2022preference} school matching, or in settings where there is no centralized matching algorithm, but  DA-like dynamics might naturally occur, as one side typically ``proposes" to the other side~\cite{echenique2024experimental}. In this paper, we introduce an endowment effect in two-sided markets by having the accepting side exhibit loyalty and study its implications.

We model loyalty as an additive parameter $k\in \{0,1\ldots,n-1\}$ where a hospital ``divorces" the current doctor they are matched with only if they get a proposal from a doctor they rank higher by at least $k+1$ positions. $k=0$ is the standard model where hospitals always choose the highest ranked doctor who proposed to them. In contrast, when $k$ equals the \textit{number of doctors minus 1}, once a hospital is matched to a doctor, they stick to this doctor until the algorithm terminates. Intuitively, loyalty should help the doctors mitigate the harmful effect of competition, as it is harder to unmatch a doctor. Thus we would not expect to see the long rejection chains that drive the lower bound of~\cite{ashlagi2017unbalanced}. However, we demonstrate that the effect of competition is much more significant than the mitigating effect of loyalty, as the average doctor rank in unbalanced markets jumps from $\Theta(\log n)$ to $\tilde{\Theta}(\sqrt{n})$ even at an extremely high loyalty $k=n-\sqrt{n}\cdot\ln n=n(1-o(1))$. This result strengthens the importance of being on the short side of a two-sided market.

\subsection{Our Results}
Without loyalty, regardless of which doctor proposes next, the algorithm always ends up in the same stable matching. In unbalanced markets, the \textit{Rural Hospital Theorem}~\cite{roth1986allocation} states that the same doctor is unmatched in all stable matchings. By introducing loyalty, multiple stable outcomes may arise, even when doctors propose. Furthermore, the unmatched doctor is no longer unique. Consequently, our analysis must be robust to any DA variant in use.

\paragraph{Balanced Market.}
In Section~\ref{sec:balanced}, we study markets where the number of doctors equals the number of hospitals. The results of this section hold for every consistent DA variant, a general form capturing many forms of loyalty. In Theorem~\ref{thm:total-prop-balanced}, we give a tight bound of $\Theta(n\log n)$ expected number of proposals for any consistent DA variant, which implies an average rank of $\Theta(\log n)$ for doctors (see Figure~\ref{fig:main_balanced}). We first provide an upper bound to the number of proposals. In a balanced market, the algorithm terminates once all hospitals receive a proposal. Similar to~\cite{wilson1972analysis,CaiT22}, we can bound the number of proposals by applying the coupon collector's bound (Corollary~\ref{cor:balanced_upperbound}). In Lemma~\ref{lem:balanced_lower_bound}, we provide a lower bound by introducing the concept of \textit{heavy} doctors, i.e., doctors that make many proposals. We differentiate between two cases: $(i)$ either the algorithm is likely to have many  heavy doctors, where we immediately  obtain an $\Omega(n\log n)$ lower bound, or $(ii)$ there are few heavy doctors, in which case the light doctors need to send out many proposals for them to get matched.

\paragraph{Unbalanced Market.}
    In Section~\ref{sec:unbalanced}, we study the case where the number of doctors exceeds the number of hospitals by $1$. When loyalty is $k=n$, no hospital will accept another proposal once $n$ hospitals are matched. Therefore,  the average rank $\Theta(\log n)$ of doctors is  maintained. In contrast, \cite{ashlagi2017unbalanced} establish that without loyalty, where $k=0$, the average rank of doctors is $\Theta(n/\log n)$. One can hope that loyalty counteracts the stark effect of competition, and some moderate amount of loyalty suffices to regain the logarithmic rank. Notably, we find that the level of loyalty needed to counteract the effect of competition is excessively high.  Our analysis reveals a phase transition, where for $k\ge n-\sqrt{n}$, the average doctors' rank is $\Theta(\log n)$, while already for $k=n-\sqrt{n}\cdot \ln n=n(1-o(1))$, the average doctors' rank is $\tilde{\Theta}(\sqrt{n})$. This phase transition is very visible in simulations (see Figure~\ref{fig:main_unbalanced}). A summary of the loyalty regimes for which we have a complete picture are presented in Figure~\ref{fig:doctor_rankings}. Our analysis mainly focuses on \textit{the unbalanced phase}, which starts once all hospitals get matched.   

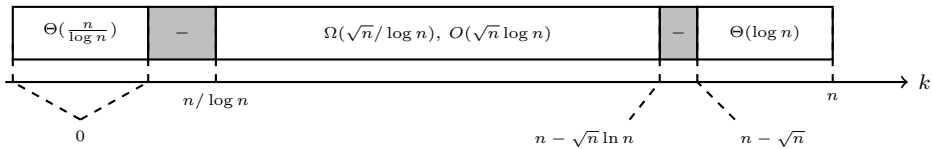
\begin{figure}[ht]
\centering
\begin{tikzpicture}[scale=1, thick]
    \node[above, align=center, font=\bfseries, font=\bfseries] at (6, 1.5) {}; 

    \draw[->] (0,0) -- (12,0) node[right, font=\scriptsize] {$k$};

    \fill[white] (0.25,0.3) rectangle (1.9,1); 
    \draw (0.1,0.3) rectangle (1.9,1); 
    \node[font=\tiny] at (1.0,0.65) {$\Theta(\frac{n}{\log n})$}; 

    \fill[gray!50] (1.9,0.3) rectangle (2.8,1); 
    \draw (1.9,0.3) rectangle (2.8,1);
    \node[font=\tiny] at (2.35,0.65) {$-$};
    
    \fill[gray!50] (8.7,0.3) rectangle (9.2,1); 
    \draw (8.7,0.3) rectangle (9.2,1);
    \node[font=\tiny] at (8.95,0.65) {$-$};
    
    \fill[white] (2.8,0.3) rectangle (8.7,1); 
    \draw (2.8,0.3) rectangle (8.7,1);
    \node[font=\tiny] at (5.75,0.65) {$\Omega(\sqrt{n}/\log n), \ O(\sqrt{n} \log n)$};

    \fill[white] (9.2,0.3) rectangle (11,1); 
    \draw (9.2,0.3) rectangle (11,1);
    \node[font=\tiny] at (10.1,0.65) {$\Theta(\log n)$};

    \draw[dashed] (0.1,0) -- (0.1,0.5); 
    \draw[dashed] (1,-0.5) -- (0.1,0);
    \draw[dashed] (1,-0.5) -- (1.9,0);
    \draw[dashed] (1.9,0) -- (1.9,1); 
    \draw[dashed] (2.8,0) -- (2.8,1); 
    \draw[dashed] (8.7,0) -- (8.7,1); 
    \draw[dashed] (8.3,-0.5) -- (8.7,0); 
    \draw[dashed] (9.7,-0.5) -- (9.2,0); 
    \draw[dashed] (9.2,0) -- (9.2,1); 
    \draw[dashed] (11,0) -- (11,1); 

    \node[below, font=\tiny] at (1,-0.5) {$0$};
    \node[below, font=\tiny] at (2.8,0) {$n/\log n$};
    \node[below, font=\tiny] at (7.7,-0.5) {$n - \sqrt{n}\ln n$}; 
    \node[below, font=\tiny] at (10.2,-0.5) {$n - \sqrt{n}$}; 
    \node[below, font=\tiny] at (11,0) {$n$};
\end{tikzpicture}
\vspace{-1ex} 
    \caption{Average doctors' rank in unbalanced market with varying level of loyalty $k$.}
    \label{fig:doctor_rankings}
\end{figure}

In Theorem~\ref{thm:very_high_loyalty_unbalanced} we provide an upper bound of $O(n\log n)$ proposals when $k\ge n-\sqrt{n}$. We show that in this case, only a few doctors propose during the unbalanced phase, as it is likely that a doctor gets rejected by all hospitals at such a high loyalty. In Theorem~\ref{thm:upper_bound_unbalanced_small_loyalty}, we obtain an  $O(n^{3/2}\log n)$ upper bound on the expected number of proposals for $k\ge  n/\log n$. Our analysis is intricate, and distinguishes between different sets of hospitals which are classified by the current rank they have of their match. We follow the progression of hospitals among these sets, and bound the number of proposals it takes for hospitals to either become unavailable for a re-match, or become very hard to match. Our analysis makes use of a variant of the {Coupon Collector Problem} which we term the \textit{Absent-Minded Coupon Collector Problem.} 

Theorem~\ref{thm:lower_bound_unbalanced} is our main and most technically involved result, where we prove that for loyalty levels $k\in[\sqrt{n}\ln n, n-\sqrt{n}\ln n]$, the expected number of proposals is $\Omega(n^{3/2}/\log n)$. Our analysis identifies a set of hospitals $T$ which \textit{(i)} is large enough and easy enough to match, so it is very unlikely for the algorithm to terminate before many of these hospitals are re-matched. \textit{(ii)} This set is hard enough to match so that the doctors need to send out many proposals in order to re-match many hospitals in $T$. These two properties drive our lower bound.

\subsection{Related Works}\label{sec:related}

The foundation of stable matching theory was laid by Gale and Shapley in their seminal work \cite{gale1962college}, which introduced the Deferred Acceptance (DA) algorithm and established key properties of stable matchings. Building on this framework, many subsequent studies have investigated one-to-one matching markets under random preferences. A significant early contribution by \citet{wilson1972analysis} examined the mathematical structure of stable matchings and the number of proposals required to achieve stability. Later, Pittel showed that the expected number of stable matchings grows quasi-linearly with market size \cite{Pittel89}, but that average rank follow a ``law of hyperbola'' \cite{pittel1992likely}, clarifying how matchings that are optimal for one side can differ significantly from matchings optimal for the other.

An important thread of this literature focuses on market imbalance in the presence of full-length preference lists. \citet{ashlagi2017unbalanced} revealed that, when the market becomes unbalanced, the advantage for the proposing side disappears because the number of potential stable matchings shrinks drastically. Building on this, \citet{CaiT22} continued to investigate how adding an agent on the proposing side transforms the matching outcome, further illuminating the dynamics of competitive imbalance. 

Another related line of research concerns ties and truncated preference lists. \citet{mcvitie1971stable} introduced the idea of ties in preference lists, allowing agents to rank different potential partners equally. \citet{ImmorlicaM05,ImmorlicaM15} studied the effects of shorter preference lists, where each agent ranks only a subset of the agents on the opposite side. Similar to imbalance, shorter lists reduce the total number of stable matchings; both unbalanced markets and truncated lists tend to favor one side and lengthen the rejection process. Nevertheless, these two scenarios differ in the intensity and nature of the competition, as well as in how the matching algorithm's runtime is affected.

\citet{KanoriaMQ21} examined the interplay between imbalance and truncated lists and found that shorter lists can sometimes mitigate the negative effects of competition caused by market imbalance. This suggests promising directions for future research: for instance, exploring how the introduction of loyalty might interact with truncated preferences to stabilize outcomes in otherwise highly competitive settings.

\begin{figure}[ht]
	\centering
	\begin{subfigure}{0.49\columnwidth}	
		\centering
		\includegraphics[width=\columnwidth]{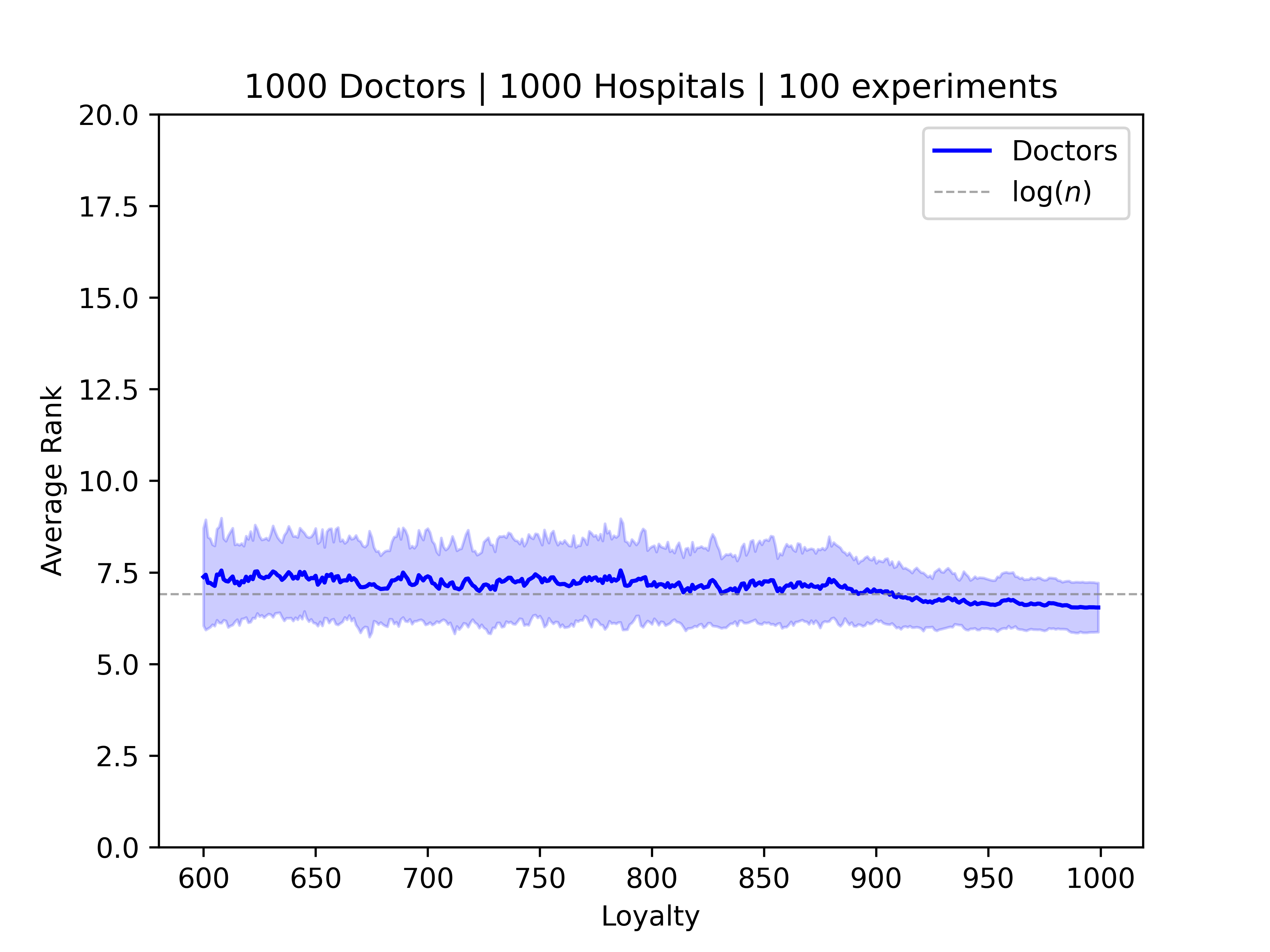}
  \vspace{-4ex}
		\caption{Average doctors' rank in a balanced market. \label{fig:main_balanced}}
	\end{subfigure}
    \begin{subfigure}{0.49\columnwidth}	
		\centering
		\includegraphics[width=\columnwidth]{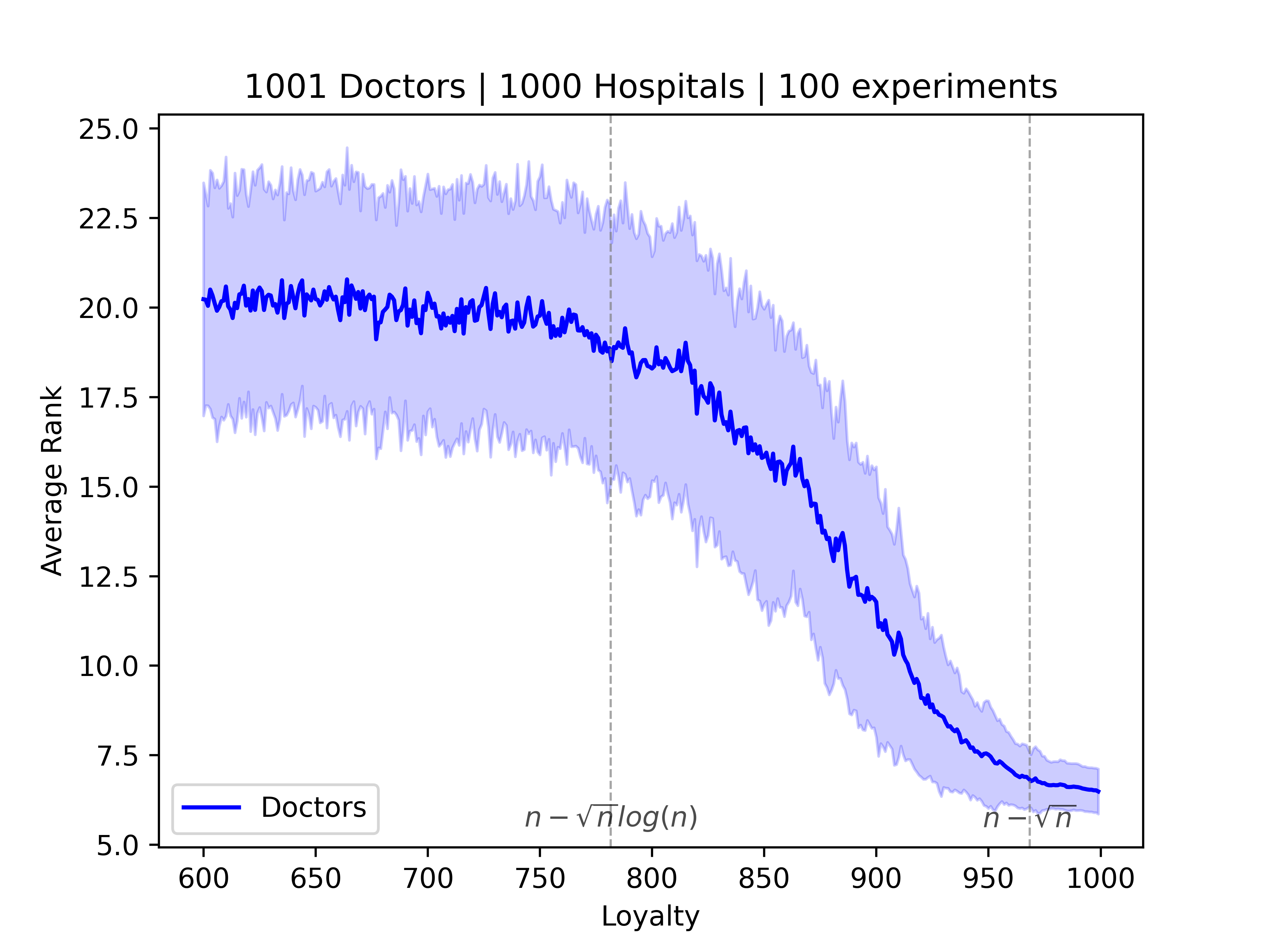}
    \vspace{-4ex}
		\caption{Average doctors' rank in an unbalanced market.\label{fig:main_unbalanced}}
	\end{subfigure}
    
        \vspace{-2ex}
	\caption{Simulations illustrating the average rank of doctors in balanced and unbalanced markets under varying levels of loyalty. On Figure~\ref{fig:main_balanced}, the horizontal dashed gray line marks the $\ln(1000)\approx 6.9$ height. On Figure~\ref{fig:main_unbalanced}, the two vertical dashed gray lines represent the two loyalties between which the phase transition occurs: $n-\sqrt{n}\ln n$ and $n-\sqrt{n}$.}
  \vspace{-2ex}
\end{figure}

\subsection{Organization}

We provide the fundamental definitions and introduce some basic tools in Section~\ref{sec:prelims}. We analyze the balanced market in Section~\ref{sec:balanced}, and the unbalanced market in Section~\ref{sec:unbalanced}. In Section~\ref{sec:simulations} we present simulations that illustrate and support our theoretical foundings. We end with a discussion in Section~\ref{sec:discussion}.

\section{Preliminaries} \label{sec:prelims}

We study a two-sided one-to-one matching market with a set of doctors $\doctors$ and a set of hospitals $\hospitals$ who are set to match. Each doctor $d$ and hospital $h$ has strict and full preference order $\succ_d$ and $\succ_h$ over $\hospitals$ and $\doctors$, respectively.

In order to present the dynamics and the role of loyalty in the market, we present the meta-algorithm, \textit{Deferred Acceptance}  (henceforth DA) in Algorithm~\ref{alg:DA}. DA takes a set of doctors, hospitals, their preferences, a function $\pull$ for choosing the next doctor to propose and a function $\accept$ for determining whether the hospital accepts the next proposal. At each iteration, the algorithm chooses the next doctor to propose according to the $\pull$ function, this doctor proposes to their favorite hospital they haven't proposed to yet, and the hospital chooses whether to accept the proposal according to the $\accept$ function. The algorithm terminates when all doctors are matched or when some doctor is rejected by all hospitals. In the vanilla setting,  $\accept(\succ_h,d,\mu(h))=true$ if and only if $d\succ_h \mu(h)$. Moreover, for every implementation of $\pull$, the same (stable) matching is returned. In our setting, since $\accept$ is modified, $\pull$ might affect the returned matching. Moreover, it might be the case where for some hospital $h$ and two doctors $d\ne d'$, both $\accept(\succ_h, d,d')=false$ and $\accept(\succ_h, d',d)=false$

\begin{algorithm}[ht!]
\caption{(Doctor-Proposing) Deferred Acceptance Meta-Algorithm (DA)}
\label{alg:DA}
\begin{algorithmic}[1]
\REQUIRE \begin{itemize}
    \item Doctors $\doctors$, Hospitals $\hospitals$.
    \item Preference lists $\{\succ_d\}_{d\in\doctors}$,$\{\succ_h\}_{h\in\hospitals}$.
    \item A function $\pull$ to choose the next doctor to propose.
    \item A function $\accept$ to determine whether the hospital accepts the proposal.
\end{itemize} 
\STATE Initialize matching $\mu$ to be empty (i.e., every agent's partner is undefined)
\STATE Initialize history $\hist \gets \phi$ to be empty.
\STATE Initialize $\mathcal{U} \gets \doctors$ to be the set of all unmatched doctors

\WHILE{$|\mathcal{U}| > 0$ and no $d\in \doctors$ was rejected by all hospitals}
  \STATE Choose $d \gets \pull(\mathcal{U},\hist)$
  \STATE Let $d$ propose to his most preferred hospital $h$ to whom $d$ has not yet proposed
  \IF{$\accept(\succ_h,d, \mu(h))=true$}
    \IF{$\mu(h)$ is defined}
      \STATE Add $\mu(h)$ to $\mathcal{U}$
    \ENDIF
    \STATE Remove $d$ from $\mathcal{U}$
    \STATE Assign $\mu(h) \gets d$
  \ENDIF
  \STATE Add $\langle d,h, \mu(h),\accept(\succ_h,d, \mu(h))\rangle$ to $\hist$ 
\ENDWHILE

\end{algorithmic}
\end{algorithm}

In order to reason about stable outcomes in the generalized DA Meta-Algorithm, we extend the notion of stability.
\begin{definition}[Stable Matching]
    Given a matching $\mu$, a doctor-hospital pair $(d,h)$ is called a \textbf{blocking pair} if $h\succ_d \mu(d)$ and $\accept(\succ_h, d, \mu(h))=true$. A  matching $\mu$ is \textbf{stable} if no doctor-hospital pair $(d,h)$ forms a blocking pair. 
\end{definition}

For DA to return a stable matching we need the following consistency property of the $\accept$ function.
\begin{definition}[Consistency]
    An $\accept$ function is consistent if for every $\succ_h$, the following holds:
    \begin{enumerate}
        \item If $d\succ_h d'$, then $\accept(\succ_h, d', d)=false$.
        \item If $\accept(\succ_h, d, d')=false$ for some $d,d'$, then for every $\hat{d}$ such that $d\succ_h \hat{d}$, $\accept(\succ_h, \hat{d}, d')=false$. 
        \item For every $h,d$, $\accept(\succ_h, d, \phi)=true$.
    \end{enumerate}
\end{definition}

Next, we establish that consistency suffices for the DA Meta-Algorithm to end up in a stable matching. 

\begin{lemma} \label{lem:stable_lemma}
    When given a consistent $\accept$ function as input, the DA Meta-Algorithm produces a stable matching.
\end{lemma}
\begin{proof}
    Since in each iteration, the proposing doctor goes down their preference list, the algorithm is guaranteed to terminate. We establish that the produced matching is stable.

    Assume towards a contradiction that there exists a blocking pair $(d, h)$ in the final matching $\mu$. Thus, $h \succ_d \mu(d)$ and $\accept(\succ_h, d, \mu(h)) = true$. By the algorithm's definition, for this to happen, a doctor $d$ must propose to $h$, and either
    \begin{enumerate}
        \item Get rejected because at the time $d$ proposes, $h$ is matched to some doctor $d'$ for which $\accept(\succ_h, d, d') = false$. By the first consistency property, $h$ can only accept proposals from doctors they prefer to their current match. Therefore, either once the algorithm terminates, $\mu(h)=d'$, and then $\accept(\succ_h, d, \mu(h)) = false$, which contradicts $(d,h)$ being a blocking pair; or $h$ ends up being matched to $\mu(h)=\hat{d}$, and $\hat{d}\succ_h d'$, and by the second consistency property, if  $\accept(\succ_h, d, d') = false$ then $\accept(\succ_h, d, \hat{d}=\mu(d)) = false$, which again contradicts $(d,h)$ being a blocking pair.
        \item $d$ gets accepted by $h$, and then gets thrown back to the market since $h$ accepted another proposal. Here, again, by the first consistency property, $h$ ends up with a doctor $\mu(h)=d'$ such that $d'\succ_h d,$ which again by the first consistency property implies that $\accept(\succ_h, d, \mu(h)) = false$, which contradicts $(d,h)$ being a blocking pair.
    \end{enumerate}
\end{proof}

\paragraph{Loyalty.} We use the $\accept$ function to model loyalty. The rank of doctor $d$ for hospital $h$ is defined as 
\begin{eqnarray}
    \rank_d(h) = 1+|\{h'\in \hospitals\ :\ h'\succ_d h\}|. 
\end{eqnarray}
Similarly, the rank of hospital $h$ for doctor $d$ is defined as
\begin{eqnarray}
    \rank_h(d) = 1+|\{d'\in \doctors\ :\ d'\succ_h s\}|. 
\end{eqnarray}
Loyalty is modeled by a parameter $k\in \{0,1,\ldots, |\doctors|-1\}$. We say the hospitals exhibit loyalty $k$ if for every $d,d',\succ_h$, 
\begin{eqnarray}
    &\accept(\succ_h, d, d')=true &\nonumber \\
    &\iff&\nonumber \\ 
   &\rank_h(d)  < \rank_h(d')-k\mbox{ or } \mu(d)=\phi.& \label{eq:loyalty}    
\end{eqnarray}

\begin{observation}
    For any loyalty parameter $k \in \{0,1,\ldots, |\doctors|-1\}$, defining $\accept$ as in Eq.~\eqref{eq:loyalty} yields a consistent $\accept$ function. Thus Algorithm~\ref{alg:DA} ends up in a stable matching for any $\pull$ function.
\end{observation}

\paragraph{Preference Generation.} Preferences orders of each doctor and hospitals are determined by an independently drawing a uniformly at random permutation over $\hospitals$ and $\doctors$, respectively.

\paragraph{Coupon Collector's Problem}
Consider a coupon collector who wants to collect $n$ types of coupons. coupons are collected one at a time, and each coupon is equally likely to be any one of $n$ types. A  folklore result shows the following.
\begin{lemma} \label{lem:coupon_collector_bound}
    Let $T$ be the random variable denoting the number of attempts required for the collector to collect the entire set of coupons. Then: $\E[T]=n\cdot H_n=\Theta(n\log n),$  and $\Pr[T>2n\ln n]\le 1/n.$
\end{lemma}

\paragraph{Amnesiac Doctors and the Coupon Collector.} 
\citet{wilson1972analysis} was the first to draw parallels between deferred acceptance in balanced matching markets (without loyalty) and the coupon collector's problem. The correspondence is the following: Generating a uniformly random preference order can be modeled by iteratively selecting the next most preferred partner uniformly at random (u.a.r.) from the list of unranked partners. Rather than generating the entire preference list upfront, whenever a doctor proposes, the next hospital to be proposed to is chosen u.a.r. from the hospitals the doctor is yet to propose to. Notice that in a balanced market, DA terminates when the $n$th hospital receives its first proposal.

To derive an upper bound on the number of proposals, we consider \textit{amnesiac doctors}. When an amnesiac doctor proposes, they choose a random hospital \textit{from the list of all hospitals}, including ones the doctor already proposed to. If the chosen hospital has already received a proposal from the doctor, it immediately rejects the proposal. The doctor repeats the process until a new hospital is selected. Since this process includes redundant proposals, it provides an upper bound on the number of proposals needed. This modified process is equivalent to the coupon collector problem, where each proposal corresponds to drawing a random coupon. Thus, the expected number of proposals can be shown to be $n \cdot H_n = O(n \log n)$. 

\subsection{Probabilistic Facts}
The following is a known property of uniform sampling. 
\begin{claim} \label{clm:min_element}
    Let $X$ denote a uniformly randomly drawn subset of $k$ distinct elements of ${1, 2, \ldots , n}$, and let $\min(X)$ denote the smallest element of $X$. We have $\E[\min(X)] = \frac{n + 1}{k + 1}$. 
\end{claim}
A proof to Claim~\ref{clm:min_element} can be found in Appendix A of \cite{CaiT22}. 


\begin{lemma}[Multiplicative Chernoff Bound]
Let $X_1, \ldots, X_n$ be independent random variables such that $\forall i \in [n], X_i \in [0, 1]$. Define $X = \sum_{i=1}^n X_i$. Then, for any $0 < \epsilon < 1$, 

$$\Pr[|X - \E[X]| > \epsilon \cdot \E[X]] \leq 2 \cdot \exp\left(-\frac{\epsilon^2 \cdot \E[X]}{3}\right).$$

\end{lemma}

\section{Balanced Market} \label{sec:balanced}

In this section, we study the case where the number of doctors equals the number of hospitals. \textit{In a balanced market, the algorithm terminates once all hospitals receive a proposal.} We show that for every variant of the DA algorithm which takes a consistent $\accept$ function as input, the expected number of proposals is $\Theta(n\log n)$. From now on, we call a DA variant with a consistent $\accept$ function, \textit{a consistent DA variant}.

\begin{theorem}
    For a consistent DA variant, there are $\Theta (n \log (n))$ proposals in expectation. Therefore, the average rank of doctors is $\Theta(\log n)$.   \label{thm:total-prop-balanced}
\end{theorem}

This matches the bound of the standard setting, where $\accept(\succ_h,d,d')=true$ simply if $d\succ_h d'$. 
Unlike the standard setting, the DA Meta-Algorithm can output multiple stable outcomes, even for the same $\accept$ function, for different $\pull$ functions.\footnote{Consider the case where there are two doctors, $d_1,d_2$ and two hospitals $h_1,h_2$, and for every doctor $d_i$, $h_1\succ_{d_i} h_2$. With loyalty $k=1$, the doctor that proposes first ends up matched to $h_1$.} We prove Theorem~\ref{thm:total-prop-balanced} by providing an upper bound in Corollary~\ref{cor:balanced_upperbound} and a lower bound in Lemma~\ref{lem:balanced_lower_bound}.


\paragraph{Upper Bound.}
As discussed in Section~\ref{sec:prelims}, for the standard DA implementation, \cite{wilson1972analysis} established an upper bound on the  expected number of proposals by analyzing the case where doctors are \textit{amnesiac}, and whenever they need to send out their next proposal, they choose a hospital uniformly at random from $\hospitals$. We notice that this reasoning holds for every consistent DA variant. 

Consider such a DA variant $A$. Fix $A$, and make the doctors amnesiac. Let $\tilde{A}$ denote the resulting algorithm. 
Let $P_i^A$ and $P_i^{\tilde{A}}$ denote the number of proposals are made between the $(i-1)$th hospital is matched until the $i$th hospital is matched in $A$ and $\tilde{A}$, respectively. Let $P^A=\sum_{i=1}^n P_i^A$  and $P^{\tilde{A}}= \sum_{i=1}^n P_i^{\tilde{A}}$ denote the total number of proposals made by algorithm $A$ and $\tilde{A}$, respectively.  Clearly, as in $\tilde{A}$ there can be redundant proposals, $\E[P_i^A]\le \E[P_i^{\tilde{A}}]$ for every $i\in [n]$, which implies $\E[P^A]\le\E[P^{\tilde{A}}]$. Thus, we only need to bound  $\E[P^{\tilde{A}}]$ to obtain our bound.

The following claim 
follows standard coupon collector arguments. The proof appears in Appendix~\ref{sec:balanced_proofs} for completeness.

\begin{observation} \label{obs:balanced_amnesiac_upperbound}
    For every amnesiac counterpart $\tilde{A}$ of some consistent DA variant $A$,  the expected number of proposals is $\E[P^{\tilde{A}}] = n\cdot H_n=O(n\log n)$.
\end{observation}

As an immediate corollary, we have the following. 
\begin{corollary}
 \label{cor:balanced_upperbound}
    For every consistent DA variant $A$, there are $O(n\log n)$ proposals in expectation.
\end{corollary}

\paragraph{Lower Bound.}
The above upper bound is not surprising, since making it harder for hospitals to reject a doctor they are currently matched with should help doctors get a better outcome. However, even though some forms of $\accept$ functions might seem to help doctors (e.g., by introducing loyalty),  we illustrate that asymptotically, this affect is insignificant. 

\begin{lemma}~\label{lem:balanced_lower_bound}
    For every consistent DA variant $A$, there are $\Omega(n \log n)$ proposals in expectation.
\end{lemma}

Consider a scenario with $t$ matched hospitals, implying $n-t$ unmatched hospitals. For a doctor $d$, let $m_d$ be the number of hospitals $d$ has already proposed to. When a doctor $d$ proposes, the probability that their proposal is directed to an unmatched hospital is $\frac{n-t}{n-m_d}$, which increases with $m_d$. 
We call a doctor $d$ with $m_d\ge\frac{n}{2}$ a \textit{heavy doctor}. 
We call non-heavy doctors \textit{light doctors}. 

Let $\varepsilon$ denote the event where there are at most $\ln n$ heavy doctors when the algorithm terminates. We inspect two cases. First, we show the bound is obtained if the probability of this event is low.

\begin{lemma} \label{lem:balanced_lb_heavy}
    Consider a consistent DA variant $A$. If $\Pr[\varepsilon]<1/2$ then there are $\Omega(n \log n)$ proposals in expectation.
\end{lemma}
\begin{proof}
    \begin{eqnarray*}
    \E[P^A] & = & \E[P^A\ | \ \varepsilon]\cdot\Pr[\varepsilon]+ \E[P^A\ | \ \neg\varepsilon]\cdot\Pr[\neg\varepsilon]\\
    &\ge & \ln n\cdot\frac{ n}{2}\cdot\Pr[\neg\varepsilon]\\
    &\ge & n\ln n/4  =  \Omega(n\log n).
\end{eqnarray*} 
\end{proof}

We next show that if this probability is high, our bound is obtained as well.

\begin{lemma} \label{lem:balanced_lb_light}
    Consider a consistent DA variant $A$. If $\Pr[\varepsilon]\ge 1/2$ then there are $\Omega(n \log n)$ proposals in expectation.
\end{lemma}
\begin{proof}
    We prove that the desired bound is obtained only by considering proposals sent out by light doctors. For $i\in [n]$, let $L_i$ be the following random variable:
    \begin{itemize}
        \item Let $S_1$ be the set of light doctors at the start of the algorithm (the set of all doctors). $L_1$ is the expected number of proposals it takes to match the first doctor in $S_1$.
        \item For $i\in [2,n]$, let $S_i$ be the set of light doctors right after $L_{i-1}$ was set (which is also a random variable). Clearly $S_i\subseteq S_{i-1}$. Some of the doctors in $S_i$ are already matched. $L_i$ is the number of proposals that is needed to increase the number of matched doctors from $S_i$ by one. If $S_i=\emptyset$, then $L_i=0$. 
    \end{itemize}

    Let $\varepsilon_i$ be the event that there are at most $\ln n$ heavy doctors at the point we start to count the proposals for $L_i$. Since $\varepsilon$ implies $\varepsilon_i$, 
    \begin{eqnarray}
        \Pr[\varepsilon_i]\ge \Pr[\varepsilon]\ge 1/2. \label{eq:epsi_bound}
    \end{eqnarray}

    We now analyze $\E\left[L_i\ |\ \varepsilon_i\right]$ for $i\le n-\ln n$. The event we are trying to analyze is the following: There are at most $\ln n$ heavy doctors (some are matched, some are unmatched), and at least $i-1$ matched hospitals. We try increasing the number of matched doctors from the set $S_i$, all of which are currently light. We note that for $i\le n-\ln n$, $S_i$ is non-empty under $\varepsilon_i$. In order for the number of doctors from $S_i$ that are matched to increase, one of the following has to happen: 
    \begin{enumerate}
        \item Either a doctor from $S_i$ proposes to an unmatched hospital. There are at most $n-i+1$ such hospitals at this point.
        \item Or a doctor from $S_i$ proposes to a hospital matched to a doctor that was heavy at the point we started to count proposals for $L_i$. There are at most $\ln n$ such hospitals under event $\varepsilon_i$.
    \end{enumerate}
    Consider some doctor proposing the $\ell$th proposal in $L_i$. If this doctor is not in $S_i$, then clearly the number of matched doctors from $S_i$ does not increase. Consider a doctor $d\in S_i$. After the $\ell$th proposal in $L_i$, we have that $m_d\le n/2+\ell-1$, since $m_d$ was less than $n/2$ at the point we started to count proposals in $L_i$ (as $d$ was light at that point). The probability of the $\ell$th proposal results in an increase in the number of matched doctors from $S_i$ is at most $$\frac{n-i+1+\ln n}{n-m_d}< \frac{n-i+1+\ln n}{n/2-\ell+1}.$$ Consider the following experiment. There are $m=n/2$ balls in a hat, of which there are $r=n-i+1+\ln n$ red balls. We now choose a ball at random, and take it out, until we choose a red ball. The expected number of trials to draw a red ball equals the expected position of the first red ball in a random permutation of the balls, given by $\frac{m+1}{r+1}$. The probability of choosing the first red ball at the $\ell$th draw is $$\frac{r}{m-\ell+1}= \frac{n-i+1+\ln n}{n/2-\ell+1},$$ which is higher than the probability the $\ell$th proposal will increase the number of matched doctors in $S_i$. Therefore, 
    \begin{eqnarray}
        \E\left[L_i\ |\ \varepsilon_i\right]> \frac{m+1}{r+1}= \frac{n/2+1}{n-i+\ln n+2}. \label{eq:li_bound}
    \end{eqnarray}
    We note that this bound also makes sense if $r>m,$ as there is at least one proposal under $\varepsilon_i$ for those values of $i$ where this occurs. We conclude:
     \begin{eqnarray*}
        \E[P^A]  & \ge & \E\left[\sum_i L_i\right]  =  \sum_i \E\left[L_i\right]\\
        &\ge & \sum_i \E\left[L_i\ |\ \varepsilon_i\right]\Pr[\varepsilon_i]
        \ge  \frac{1}{2} \sum_{i\in[n-\ln n]} \E\left[L_i\ |\ \varepsilon_i\right]\\
        &> &\frac{1}{2} \sum_{i\in[n-\ln n]} \frac{n/2+1}{n-i+\ln n+2}  > \frac{n}{4} \sum_{i= \ln n+2}^{n+\ln n+1}\frac{1}{i} \\
        & =& \frac{n}{4}(H_{n+\ln n + 1}-H_{\ln n +1})  >  \frac{n\ln n}{4}(1-o(1)) \\
        & =&  \Omega(n\log n).
    \end{eqnarray*}
    Here, the third inequality follows Eq.~\eqref{eq:epsi_bound} and the fourth inequality follows Eq.~\eqref{eq:li_bound}.
\end{proof}

Lemma~\ref{lem:balanced_lower_bound} immediately follows Lemma~\ref{lem:balanced_lb_heavy} and Lemma~\ref{lem:balanced_lb_light}.



\section{Unbalanced Market} \label{sec:unbalanced} 

We next turn to the unbalanced case, where the number of doctors exceeds the number of hospitals by one. We focus on the case where hospitals exhibit loyalty, as defined in Eq.~\eqref{eq:loyalty} with parameter $k\in\{0,\ldots,\doctors-1\}$. We call such a DA variant, an \textit{L-DA variant}. In the setting without loyalty ($k=0$), \cite{ashlagi2017unbalanced} show that there are $\Theta(n^2/\log n)$ proposals in expectation. This implies that the average rank of doctors is $\Theta\left(n/\log n\right)$. On the other hand, when $k=n$, meaning once a doctor is matched, they remain matched, the total number of proposals is at most the number needed to match $n$ hospitals, which was shown to be $\Theta(n\log n)$, plus at most an additional $n$ proposals for the unmatched doctor. This implies $\Theta(\log n)$ average doctors' rank. 

This leads us to the main question of this section: \textit{How much loyalty is needed in order to ensure a poly-logarithmic average doctors' rank?} Interestingly, we prove that an excessively high level of loyalty is required. Specifically, with loyalty $k = n - \sqrt{n}\cdot\log n=n(1-o(1))$, the average rank of doctors is as high as $\tilde{\Theta}(\sqrt{n})$.

In our analysis, we consider two phases:
\begin{enumerate}
    \item \textit{The balanced phase:} Spanning all proposals until the moment when $n$ doctors and $n$ hospitals are tentatively matched.
    \item \textit{The unbalanced phase:} Beginning after $n$ doctors and $n$ hospitals are tentatively matched and ending when a doctor is rejected by all $n$ hospitals, causing the algorithm to terminate.
\end{enumerate}

The number of proposals during the balanced phase follows the same bounds as in the balanced market. The proof is omitted as it is identical to the proof of Theorem~\ref{thm:total-prop-balanced}. The high probability bound on the number of proposals follows Lemma~\ref{lem:coupon_collector_bound}.

\begin{lemma} \label{lem:balanced_phase_bound}
    For every L-DA variant and every loyalty parameter $k$, the expected number of proposals during the balanced phase is $\Theta(n\log n)$. The probability that the number of proposals is higher than $2n\ln n$ is at most $1/n$.
\end{lemma}

Thus, the number of proposals during the unbalanced phase is the key factor determining the doctors' rank, as demonstrated in simulations in Figure~\ref{fig:total_both_phases}.

Let $f_h$ be the rank of the first proposal of hospital $h$. After the balanced phase, only hospitals with $f_h > k+1$ remain available. At this stage, all DA variants converge, as there is only one unmatched doctor proposing at any given step. We note that $f_h$ is distributed uniformly at random from the set $\{1,2,\ldots, n+1\}$. 
As before, we call a doctor $d$ \textit{heavy} if $m_d\ge n/2,$ and otherwise, we call the doctor \textit{light}. We also let $m_h$ denote the number of offers received by some hospital $h$. We call hospitals for which $m_h\ge n/2$ \textit{heavy hospitals}, and conversely, non-heavy hospitals \textit{light hospitals}.

\subsection{Loyalty $> n - \sqrt{n}$}

\begin{theorem}
    For any L-DA variant with $k > n - \sqrt{n}$, the expected number of proposals is $\Theta(n \log n)$ and the expected doctors' rank is $\Theta(\log n)$. \label{thm:very_high_loyalty_unbalanced}
\end{theorem}

\begin{proof}
Let $S_A$ be the set of hospitals that are available at the start of the unbalanced phase. For a hospital to be in $S_A$, it must be that $h_{f}>k+1\Rightarrow h_f\in [n-\sqrt{n}+1,n]$. Let $\varepsilon_1$ denote the event that $|S_A| \leq 2\sqrt{n}$. Note that for every hospital, $h_f\in [n-\sqrt{n}+1,n]$ with probability $1/\sqrt{n}$, therefore, $\E[|S_A|]\le \sqrt{n}.$ Applying Chernoff bound, we get that $\Pr[\varepsilon_1] \geq 1 - 1/n$.     
    Let $\varepsilon_2$ denote the event that at most $4\ln n$ heavy hospitals exist at the end of the balanced phase. By Lemma~\ref{lem:balanced_phase_bound}, $\Pr[\varepsilon_2] \geq 1 - 1/n$.

    For a light hospital in $S_A$ to accept a proposal, the proposing doctor must be ranked among the hospital's top $\sqrt{n}$ preferences.  We next show that under event $\varepsilon_1$, if we let  $8e^5\ln n$ doctors propose to light hospitals in $S_A$ from the beginning of the unbalanced phase, the probability at most $2\ln n$ of them will be rejected by all light hospitals is at most $1/n$. This implies that at least $2\ln n+1$ of them will be rejected by all light hospitals in $S_A$ with probability $1-1/n$. Since that under even $\varepsilon_2$ there are at most $2\ln n$ heavy hospitals in $S_A$, this implies that with probability at least $1-1/n$, at least one doctor of the first $8e^5\ln n = O(\log n)$ doctors proposing in the unbalanced phase will be rejected by all hospitals in $S_A$, which yields our bound on the number of proposals.

    Consider the $\ell$th doctor proposing, where $\ell\in[8e^5\ln n]$. the probability that the doctor will be rejected by all light hospitals in $S_A$ under $\varepsilon_1$ is at least \begin{eqnarray*}
        \left(1-2/(\sqrt{n}-\ell)\right)^{|S_A|}&\ge& \left(1-2/(\sqrt{n}-8e^5\ln n)\right)^{2\sqrt{n}}\\
        &\ge &\left(1-5/2\sqrt{n}\right)^{2\sqrt{n}} \ge e^{-5},
    \end{eqnarray*}
    where we use the fact that, under $\varepsilon_1$, $|S_A|\le 2\sqrt{n}$, in the first inequality. Therefore, under $\varepsilon_2$, the expected number of doctors that will be rejected by all light hospitals in $S_A$ among the first $8e^5\ln n$ doctors to propose in the balanced phase is at least $8\ln n$. By Chernoff bound, the probability that at most $2\ln n$ doctors of the first $8e^5\ln n$ to propose in the unbalanced phase will be rejected by all light hospital in $S_A$ is at most $e^{-\ln n}=1/n.$

    Let $D_U$ be the random variable for the number of doctors proposing in the unbalanced phase. We have that 
    \begin{eqnarray*}
        \E[D_U] & = & \E[D_U |  \varepsilon_1 \wedge\varepsilon_2] \cdot \Pr[\varepsilon_1\wedge\varepsilon_2] \ + \  \E[D_U |  \neg\varepsilon_1 \vee \neg\varepsilon_2]\cdot\Pr[\neg\varepsilon_1 \vee \neg\varepsilon_2]\\
        &\le & \E[D_U |  \varepsilon_1 \wedge\varepsilon_2] + n\cdot \frac{2}{n}\\
        & \le & 8e^5\ln n\cdot(1-1/n) + n\cdot \frac{1}{n} + 2 \\
        & = &  O(\log n).
    \end{eqnarray*}
    
    Since each doctor can propose at most $n$ proposals, we get that there are $O(n\log n)$ proposals in the unbalanced phase in expectation (and with probability $1-O(1/n)$). By Lemma~\ref{lem:balanced_phase_bound}, we get that the number of proposals overall is $\Theta(n\log n)$, as desired.
\end{proof}

Figure~\ref{fig:478axes} illustrates that only a small number of hospitals in $S_A$ improve their rank during the unbalanced phase, as our analysis demonstrates.

\subsection{Loyalty $\le n-\sqrt{n}\cdot \log n$} 

\paragraph{An Upper Bound.}

We provide the following upper bound on the number of proposals for this regime.

\begin{theorem}\label{thm:upper_bound_unbalanced_small_loyalty}
    For any L-DA variant, when $k=n/c(n)$, for $c(n)<\sqrt{n}$, there are $$O(c(n)^2n\log n + n^{3/2}\log n + c(n)n^{3/2})$$ proposals in expectation. This implies that the expected rank of doctors is $$O(c(n)^2\log n + \sqrt{n}\log n + c(n)\sqrt{n}).$$
\end{theorem}

From now on, we denote $c=c(n)$. When $k=n/c$, a hospital can accept up to $c$ proposals. We assign hospitals to sets based on the rank of their current match, and analyze the progression of the hospitals through the sets. Consider the following sets:

\begin{itemize}
    \item $F = \{h\ : \ \rank_h(\mu(h))\in [n/c] \}$.
    This is the set of unavailable hospitals that have reached their \textbf{final} match, and will not accept further proposals.
    \item $H = \{h :  \ \rank_h(\mu(h))\in [n/c+1,n/c+\sqrt{n}] \}$. This is the set of \textbf{hard-to-match} hospitals. We bound the number of hospitals that belong to this set at some point of the unbalanced phase.
    \item $M= \{h\ : \ \rank_h(\mu(h))\in [n/c+\sqrt{n}+1, 2n/c]\}$. This is the set of hospitals that are \textbf{moderately} hard to match. Each $h\in M$ accept a proposal from their top $\sqrt{n}$ doctors. Thus, when a hospital in $M$ receives a proposal, the probability of acceptance is at least $\sqrt{n}/n=1/\sqrt{n}$.  
    \item For $i\in [c-2]$, $$E_i = \{h\ : \ \rank_h(\mu(h))\in [(i+1)n/c+1, (i+2)n/c]\}.$$ A hospital in $E_i$ can be matched to their top $in/c$ doctors; therefore, such a hospital accepts a proposal with probability at least $(in/c)/n=i/c$.
    \item $U= \{h\ :\ h\mbox{ is unmatched}\}$. A hospital in $U$ will accept a proposal w.p. 1. 
\end{itemize}

We bound the number of proposals it takes to move hospitals from $\{E_i\}_i\cup\{U,M\}$ to $\{H,F\}$ in two phases. 
\begin{enumerate}
    \item \textit{The easy phase:} We first bound the expected number of proposals it takes to move all hospitals from $\{E_i\}_i\cup\{U\}$ to $\{M,H,F\}$.
    \item \textit{The moderate phase:} Then, assuming there are no hospitals in $\{E_i\}\cup\{U\}$, we bound the expected number of proposals to move hospitals from $M$ to $\{H,F\}$.
\end{enumerate}

To bound the number of proposals in each phase, we first analyze how many proposals would be required when doctors are amnesiac. Since there are redundant proposal, this gives an upper bound on the actual number of proposals. We again use a coupon-collector-style argument to analyze this case. However, in this scenario, a hospital that receives a proposal may not accept it, as it could already be matched. Nevertheless, in both phases, hospitals in the relevant sets accept proposals with sufficiently high probability. 
We define the following variant of the coupon collector problem.

\begin{definition}[Absent-minded coupon collector]
    The Absent-minded coupon collector problem is defined by two parameters, $n$ and $q\in (0,1]$. As before, $n$ is the number of coupon types to be collected. Unlike the vanilla coupon collector problem, here once the collector collects a coupon, they might forget where they kept the coupon. The collector remembers where the coupon is kept with probability at least $q$. The goal of the collector is to be able to keep all types of coupons. 
\end{definition}

We have the following bound on the expected number of attempts required for the collector to complete the entire set.

\begin{lemma}~\label{lem:absent_minded}
    In the absent-minded coupon collector's problem, the expected number of attempts required to complete the collection of all coupon types is at most $\frac{n}{q} \cdot  H_n$.
\end{lemma}

\begin{proof}
After collecting $i$ coupons, the probability of obtaining a new coupon on the next attempt is at least $\frac{n-i}{n} \cdot q$. The expected number of attempts needed to collect the $i+1$th coupon is at most $\frac{n}{(n-i)q}$. Therefore, to collect all coupons, the total expected number of attempts is at most: $$\sum_{i=0}^{n-1} \frac{n}{(n-i)q} = \frac{n}{q} \sum_{i=0}^{n-1} \frac{1}{i} = \frac{n}{q}  \cdot H_n.$$
\end{proof}

Using the above lemma to prove the following two lemmas.

\begin{lemma}~\label{lem:ub_easy_phase}
    The expected number of proposal during the easy phase is $O(c^2n\log n)$.
\end{lemma}
\begin{proof}
    Each hospital starts in set $U$ and transitions through at most $c-1$ sets before ending in $H$ or $F$. This process can be modeled as collecting $n(c-1)$ coupons, where each coupon corresponds to a hospital moving to a new set. As mentioned above, the hospitals accepts proposals with probability at least $1/c$. Thus, we can obtain an upper bound on the expected number of proposals during the easy phase by applying Lemma~\ref{lem:absent_minded} with $n(c-1)$ coupons and $q=1/c$. The expected number of proposals during the easy phase is thus bounded by:
    $$c(c-1)nH_{(c-1)n}=O(c^2 n \log n)$$ proposals.
\end{proof}

\begin{lemma}~\label{lem:ub_medium_phase}
    The expected number of proposal during the moderate phase is $O(n^{3/2}\log n)$.
\end{lemma}
\begin{proof}
    Each hospital in $M$ accepts a proposal with probability at least $1/\sqrt{n}$, as these hospitals will accept a proposal from their top $\sqrt{n}$ doctors. Hospitals already in $H$ or $F$ can be treated as "collected", and do not require further proposals. Thus, we can apply Lemma~\ref{lem:absent_minded} again with $n$ coupons and $q=1/\sqrt{n}$, and immediately get the bound in the lemma statement.
\end{proof}

We bound the expected number of hard-to-match hospitals.
\begin{lemma} \label{lem:H_bound}
    The probability that more than $2c\sqrt{n}$ hospitals belong to $H$ after the moderate phase is at most $1/n$.
\end{lemma}
\begin{proof}
    We notice that whenever a proposal can move a hospital to set $H$, it can also move the hospital to set $F$. Therefore, the probability of a hospital to move to set $H$ at some point is at most $\sqrt{n}/(\sqrt{n} + n/c)< c/\sqrt{n}$. In expectation, there are less than $c\sqrt{n}$ such hospitals, and by Chernoff bound, the probability there are more than $2c\sqrt{n}$ such hospitals is at most $e^{-c\sqrt{n}/3}\ll 1/n$. 
\end{proof}

We use the above to prove Theorem~\ref{thm:upper_bound_unbalanced_small_loyalty}:
\begin{proof}
By Lemmas~\ref{lem:ub_easy_phase} and ~\ref{lem:ub_medium_phase}, in expectation, after $O((c^2+\sqrt{n})n\log n)$ proposals, all hospitals are in sets $H$ and $F$. Once all hospitals are in $H$ and $F$, only hospitals in $H$ can get matched once, and the number of doctors proposing is at most $|H|+1$. By Lemma~\ref{lem:H_bound}, $E[|H|]\le 2c\sqrt{n}\cdot\frac{n-1}{n} + n\cdot\frac{1}{n}=O(c\sqrt{n})$. Since each doctor can propose to at most $n$ hospitals, we get that in expectation, there are $O(cn^{3/2})$ proposals after the moderate phase. We conclude that the expected number of proposals is $O(c^2n\log n+n^{3/2}\log n+cn^{3/2})$.
\end{proof}

We get the following corollary.
\begin{corollary}~\label{cor:upper_bound_low_loyalty}
    When $k\ge n/\log n$, the expected rank for doctors in an unbalanced market is $O(\sqrt{n}\log n)$. \label{crl:upper_bound_unbalanced}
\end{corollary}
\begin{proof}
    Plugging in $c(n)=\log n$ in Theorem~\ref{thm:very_high_loyalty_unbalanced}, we get that the expected number of proposals is 
    $$O(n\log^3 n+n^{3/2}\log n)=O(n^{3/2}\log n),$$ implying that the expected doctors' rank is $O(\sqrt{n}\log n)$.
\end{proof}

\paragraph{A Lower Bound.}

We show the following. 

\begin{theorem} \label{thm:lower_bound_unbalanced}
    For any L-DA variant, when $k\in [\sqrt{n}\ln n ,n-\sqrt{n}\ln n]$, there are $\Omega(n^{3/2}/\log n)$ proposals in expectation and the expected doctor rank is $\Omega(\sqrt{n}/\log n)$.
\end{theorem}

Consider some loyalty $k\in [\sqrt{n}\ln n ,n-\sqrt{n}\ln n]$, and $\ell=\sqrt{n}\ln n$. Our proof focuses on hospitals that are first matched to a doctor they rank in the $[k+\ell/2+1, k+\ell]$ range. This set of hospitals is large enough for the following two thing to happen:
\begin{enumerate}
    \item They are easy enough to re-match such that a doctor is not likely to be rejected by all such hospitals; therefore, we need to re-match a large portion of these hospitals. 
    \item They are hard enough to re-match such that the doctors will need to send out $\Omega(n^{3/2}/\log n)$ proposals in order to match a large portion of these hospitals.
\end{enumerate}

Indeed, Figure~\ref{fig:361axes} suggest this is the case, as almost all hospitals in this range improve their match before the algorithm terminates.

Recall that $f_h$ is the rank of the first doctor who proposed to (and accepted by) hospital $h$. We define the following set of hospitals:
\begin{eqnarray}
T=\{h \ : \ f_h\in [k+\ell/2+1,k+\ell]\}.    
\end{eqnarray}
 
We define event $\varepsilon_1$ as the event where the following two things happen during the balanced phase:
\begin{enumerate}
    \item $|T|\in [\ell/4,\ell]$.
    \item For every doctor, the number of hospitals in $T$ they proposed to during the balanced phase is at most $\ell/8$.
\end{enumerate}

We show that this event is very likely. 

\begin{lemma} \label{lem:eps1_high}
    $\Pr[\neg \varepsilon_1] \le 1/2\ln n.$ 
\end{lemma}

To prove the lemma, we first need to bound the probability that many hospitals in $T$ get re-matched during the balanced phase.

\begin{lemma}
    The probability that more than $\ell/16$ hospitals in $T$ have been re-matched during the balanced phase is at most $32/\ln^2 n.$ \label{lem:ell_rematch_bound}
\end{lemma}
\begin{proof}
     In order for a hospital in $T$ to be re-matched, one of the proposals directed towards the hospital has to be sent from one of $h$'s $\ell$ favorite doctors. We analyze this probability for event in the version where proposing doctors are amnesiac, where reason about the number of proposal $h$ gets,  assuming that every proposal is from a distinct doctor. This is an upper bound on the actual number of proposals, thus making it easier to re-match $h$. The probability of this event is obviously higher than when doctors are not amnesiac.

    Let $\varphi_1$ be the event that no more than $2n\ln n$ proposals have been sent by amnesiac doctors during the balanced phase. By Lemma~\ref{lem:balanced_phase_bound}, the probability of this event is at least $1-1/n$. Under event $\varphi_1$  the number of hospitals who received more than $\sqrt{n}/\ln^3 n$ proposals is at most $2\sqrt{n}\ln^4 n$. 
    Therefore, under $\varphi_1$, the probability a hospital gets more than $\sqrt{n}/\ln^3 n$ proposals is at most $2\sqrt{n}\ln^4 n/n=2\ln^4 n/\sqrt{n}$. The probability that none of $t\le \sqrt{n}/\ln^3 n$ proposals will be directed to $h$ is 
    \begin{eqnarray}
        \frac{\binom{n-\ell}{t}}{\binom{n}{t}} & \ge & \frac{\binom{n-\sqrt{n}\ln n}{\sqrt{n}/\ln^3 n}}{\binom{n}{\sqrt{n}/\ln^3 n}} \nonumber \\ 
        &\ge & \frac{(n-\sqrt{n}\ln n-\sqrt{n}/\ln^3 n)^{\sqrt{n}/\ln^3 n}}{n^{\sqrt{n}/\ln^3 n}}\nonumber \\ 
        & \ge & (1-2\ln n/\sqrt{n})^{\sqrt{n}/\ln^3 n} \nonumber\\
        &\ge &1 -2/\ln^2 n. \label{eq:xh_few_bound}
    \end{eqnarray} 

    Let $\varphi_{2,h}$ be the event that $h$ got at most $\sqrt{n}/\ln^3 n$ proposals. We have that
    \begin{eqnarray}
        \Pr[\varphi_{2,h}] & \ge & \Pr[\varphi_{2,h} | \varphi_1]\cdot \Pr[\varphi_1] \nonumber\\
        & = & (1-2\ln^4 n/\sqrt{n})\cdot (1-1/n)\nonumber\\
        &\ge & 1-3\ln^4 n/\sqrt{n}. \label{eq:varphi2_bound}
    \end{eqnarray}
    Let $X_{h,k}$ be an indicator variable for the event that $h$ got a proposal from one of their top $k$ doctors during the balanced phase. We have that
    \begin{eqnarray*}
        \Pr[X_{h,k}] & = & \Pr[X_{h,k} | \varphi_{2,h}]\cdot \Pr[\varphi_{2,h}] + \Pr[X_{h,k} | \neg \varphi_{2,h}]\cdot \Pr[\neg\varphi_{2,h}] \\
        &\le & \Pr[X_{h,k} | \varphi_{2,h}] + \Pr[\neg\varphi_{2,h}]\\
        & \le & 2/\ln^2 n + 3\ln^4 n/\sqrt{n}\\
        & \le & 4/\ln^2 n,
    \end{eqnarray*}
    where the second inequality follows Eq.~\eqref{eq:xh_few_bound} and \eqref{eq:varphi2_bound}.

    Let $X_{h,T}$ be the event that $h\in T$. As $f_h$ is distributed uniformly at random, the probability that a hospital $h$ is in $T$ is $(\ell/2)/n=\ell/2n$.

    Let $X_h$ 
    be the event that $h\in T$ and $h$ got re-matched (and therefore, got a proposal from one of their top $k$ doctors) during the balanced phase. Notice that $X_{h,T}$ and $X_{h,k}$ are negatively affiliated, as when a hospital is in $T$, there is at least one proposal not from the top $k$ doctors. 
    We get that 
    \begin{eqnarray*}
        \E[X_h] & \le & \E[X_{h,T}\wedge X_{h,k}]\le \E[X_{h,T}]\cdot \E[X_{h,k}] \\
        &\le& \frac{4}{\ln^2 n} \cdot \frac{\ell}{2n} = \frac{2}{\sqrt{n}\ln n}. 
    \end{eqnarray*}

    Let $X$ be the random variable of the number of doctors in $T$ that got re-matched during the balanced phase. We have that $$\E[X] = \sum_h[X_h] \le n\cdot \frac{2}{\sqrt{n}\ln n} = 2\sqrt{n}/\ln n.$$
    By Markov inequality, we get that 
    $$\Pr[X\ge \ell/16] \le \frac{\E[X]}{\ell/16}\le \frac{2\sqrt{n}/\ln n}{\sqrt{n}\ln n/8} = 32/\ln^2 n.$$
\end{proof}

We are now ready to prove Lemma~\ref{lem:eps1_high}.
\begin{proof}[Proof of Lemma~\ref{lem:eps1_high}]
    Recall that $f_h$ is distributed uniformly at random in $[n]$. Thus, $E[|T|]=\ell/2$, and by applying Chernoff bound, we easily obtain that $\Pr[|T|\in [\ell/4,\ell]]\gg 1-1/8\ln n$.

    We now bound the probability that there exists a doctor that proposes to more than $\ell/8$ proposals in $T$ during the balanced phase. 
    
    We first bound the probability that some doctor $d$ proposed to more than $\ell/32$ matched hospitals in $T$ that are still available for a re-match. For a matched hospital in $T$ to accept a proposal, the proposal has to come from a doctor they rank among their top $\ell=\sqrt{n}\log n\ge k$ doctors. That is, after a hospital in $T$ accepts a proposal, they will be matched to this doctor until the algorithm terminates. Hence, for a doctor to propose to more than $\ell/32$ hospitals in $T$ that are available for a re-match, the doctor has to be sequentially rejected by $\ell/32$ such hospitals. Notice that an available hospital in $T$ accepts a proposal with probability of at least $\frac{\ell/2}{n}$. Thereby, a doctor is rejected by $\ell/32$ such hospitals with probability at most
    \begin{eqnarray*}
        (1-\ell/2n)^{\ell/32} & = & (1-\ln n/2\sqrt{n})^{\sqrt{n}\ln n/32}\\
        & = & \left((1-\ln n/2\sqrt{n})^{2\sqrt{n}/\ln n}\right)^{\ln^2 n/64}\\
        & < & e^{\ln^2 n/64}\\
        & = & n^{-\ln n/64}.
    \end{eqnarray*}
    Taking a union bound, the probability there exists such a doctor is at most $$n^{-\ln n/64+1}\ll 1/8\ln n.$$

    We next bound the probability that a doctor was the first matched doctor of more than $\ell/32$ hospitals in $T$.  This probability is lower than the probability that a doctor was the first match of more than $\ell/32$ hospitals {in $\hospitals$}. For a doctor to be the first match of more than $\ell/32$ hospitals, it has to be unmatched from the first $\ell/32$ {hospitals} the doctor was the first match of. Therefore, the doctor cannot be one of $k$ favorite doctors of all such hospitals. 
    Since when a doctor proposes to a hospital, their rank for the hospital is distributed uniformly at random, the probability for this is at most    
    \begin{eqnarray*}
        (1-k/n)^{\ell/32} & \le & (1-\ell/n)^{\sqrt{n}\ln n/32}\\
        & = & \left((1-\ln n/\sqrt{n})^{\sqrt{n}/\ln n}\right)^{\ln^2 n/32}\\
        & < & e^{\ln^2 n/32}\\
        & = & n^{-\ln n/32}.
    \end{eqnarray*}
    Taking a union bound, the probability there exists such a doctor is at most $$n^{-\ln n/32+1}\ll 1/8\ln n.$$

    We finally bound the probability that a doctor proposes to more than $\ell/16$ doctors in $T$ that are already re-matched. This is bounded by the probability there are $\ell/16$ re-matched doctors in $T$ during the balanced phased, which is at most $32/\ln^2 n$ according to Lemma~\ref{lem:ell_rematch_bound}. Thus, by taking a union bound, the probability there exists a doctor at the end of the balanced phase who proposed to more than $\ell/8$ hospitals in $T$ during the balanced phase is at most $$1/8\ln n+1/8\ln n+32/\ln^2\le 3/8\ln n.$$ Taking the union bound again on the probability of both conditions, we get that the probability of $\neg \varepsilon_1$ is at most $3/8\ln n+1/8\ln n=1/2\ln n$. 
\end{proof}

Let $\varepsilon_2$ be the event that when the algorithm terminates, there are at most $\sqrt{n}$ heavy doctors and at most $\sqrt{n}$ heavy hospitals. If the probability of $\varepsilon_2$ is low then we already obtain our bound.

\begin{lemma} \label{lem:lower_bound_unbalanced_not_epsilon_2}
    If $\Pr[\varepsilon_2] < 1/2$, then $\E[P^A]=\Omega(n^{3/2})$ 
\end{lemma}
\begin{proof}
    Notice that if we have either $\sqrt{n}$ heavy doctors or $\sqrt{n}$ heavy hospitals, then we already have $\Omega(n^{3/2})$ proposals. Hence,  
    $\E[P^A]\ge \E[P^A| \neg \varepsilon_2]\cdot \Pr[\neg \varepsilon_2]= \Omega(n^{3/2}).$
\end{proof}

Consider a hospital $h\in T$. We know that for this hospital, $f_h \in [k+\ell/2+1,k+\ell].$ Thus, if this hospital accepts another proposal after their first proposal from some doctor $d$, they will be matched to a doctor in their top $\ell$ doctors. We call such hospital a \textit{re-matched} hospital. Since $k\ge \sqrt{n}\log n,$ we know that after this point, $h$ and $d$ will be matched until the algorithm terminates. Under events $\varepsilon_1$ and $\varepsilon_2$, we will need to re-match many light hospitals to many light doctors, which will result in many proposals. For $i\in[\ell]$, let $L_i$ be the following random variable:
    \begin{itemize}
        \item Let $S_1$ be the set of light doctors at the start of the algorithm (which is the set of all doctors). $L_1$ is the expected number of proposals it takes to re-match a light hospital to a doctor in $S_1$ (a match which will not be un-tied). 
        \item For $i\in [2,n]$, let $S_i$ be the set of light doctors after the $i-1$th light hospital to light doctor re-matching. Clearly $S_i\subseteq S_{i-1}$. $L_i$ is the number of proposals that is needed from the point a doctor in $S_{i-1}$ is re-matched to a light hospital to the point a doctor in $S_i$ is re-matched a light hospital. 
    \end{itemize}

 Let $\varepsilon_3$ be the event that the algorithm does not terminate before $\ell/8+2\sqrt{n}$ hospitals in $T$ are re-matched.

\begin{lemma} \label{lem:eps3_high}
    $\Pr[\neg \varepsilon_{3} | \varepsilon_1 ]\le 1/2\ln n.$
\end{lemma}
\begin{proof}
    In an unbalanced market, the algorithm ends once a doctor is rejected by all hospitals. Since we are considering the case where $k\ge \ell$, once a doctor is accepted by a hospital in $T$ which is currently matched, they would be matched to this hospital until the algorithm terminates. Recall that the balanced phase ends once $n$ hospitals have been matched. Therefore, the algorithm ends once there is a doctor in the unbalanced phase that is rejected by all hospitals. Consider some doctor that needs to be rejected before $\ell/8+2\sqrt{n}$ hospitals in $T$ are re-matched. Under $\varepsilon_1$, this doctor have proposed to at most $\ell/16$ hospitals in $T$ during the balanced phase. And since there are at least $\ell/4$ hospitals in $T$ under $\varepsilon_1$, for a doctor to be rejected by all hospitals in the balanced phase, they would need to be rejected by at least $\ell/4-(\ell/8+2\sqrt{n})-\ell/16=\ell/16-2\sqrt{n}\ge \ell/17$ hospitals in $T$. The probability for a doctor to be rejected by a hospital in $T$ is at most $1-(\ell/2)/n$. Therefore, the probability the doctor will be rejected by $\ell/17$ hospitals in $T$ is at most 
    \begin{eqnarray*}
        (1-\ell/2n)^{\ell/17} & = & (1-\ln n/2\sqrt{n})^{\sqrt{n}\ln n/17}\\
        & = & \left((1-\ln n/2\sqrt{n})^{2\sqrt{n}/\ln n}\right)^{\ln^2 n/34}\\
        & < & e^{\ln^2 n/34}\\
        & = & n^{-\ln n/34}.
    \end{eqnarray*}
    By the union bound over all doctors, we get that the probability some doctor was rejected before $i\in \sqrt{n}\log n/17$ doctors are matched to hospitals in $T$ is at most $n^{-\ln n/34}\cdot n=n^{-\ln n/34+1}\ll 1/2\ln n.$
\end{proof}

Let $\varepsilon_{2,i}$ be the event that after the $i-1$th light doctor to light hospital re-matched, there are at most $\sqrt{n}$ heavy doctors and at most $\sqrt{n}$ heavy hospitals. We now analyze $\E[L_i|\varepsilon_1,\varepsilon_{2,i},\varepsilon_{3,i}]$. The intuition is the following.  We need to match a light doctor to a light hospital in $T$. Since a light doctor have not proposed to at least $n/2$ hospitals and since there are at most $\sqrt{n}\ln n$ light hospitals in $T$ under $\varepsilon_1$, light doctors need  $\Omega((n/2)/\sqrt{n}\log n)=\Omega(\sqrt{n}/\log n)$ proposals to hit a light hospital. Once hitting a light hospital $h$, in order to match to $h\in T$, they will need to rank in $h$'s top $\ell$ hospitals to have a chance to match $h$. 
Since $h$ is a light hospital, $h$ received at most $n/2$ proposals so far, and the expected number of proposals light hospitals need to receive in order to accept a proposal is $\Omega((n/2)/\ell)=\Omega(\sqrt{n}/\log n)$. Hence, light doctors will need to propose $\Omega(n/\log^2 n)$ proposals in order to match to a light hospital.    

\begin{lemma} \label{lem:li_bound}
    For every $i\in [\ell/8],$
    $$\E[L_i|\varepsilon_1,\varepsilon_{2,i},\varepsilon_{3}]\ge \frac{(n/2+1)^2}{(\sqrt{n}\log n+1)^2}= \Omega(n/\log^2 n).$$
\end{lemma}
\begin{proof}
    We analyze the following situation: There are at most $\sqrt{n}\log n$ light hospitals, and we want to give a lower bound on the expected number of proposals needed for the $i$th light doctor to light hospital match. 
    Notice that under $\varepsilon_{2,i}$ There are at most $\sqrt{n}$ light doctors and at most $\sqrt{n}$ light hospitals. Thus, at most $\ell/8+2\sqrt{n}$ hospitals from $T$ have been matched so far, and under $\varepsilon_3$, the algorithm hasn't terminated yet.
    
    Consider the $j$th proposal since the $i$th light doctor to light hospital match. the probability that a light doctor will propose to a light hospital is at most $$\frac{\sqrt{n}\log n}{m_d-j}\le\frac{\sqrt{n}\log n}{n/2-j}.$$ 
    Consider the $\hat{j}$th proposal of phase $L_i$ that was directed to a light hospital $h$. If this proposal was made by a heavy doctor, then the probability the phase will end is $0$. If this proposal was made by a light doctor $d$, then $h$ might receive this proposal only if  $h$ ranks $d$ as one of their top $\hat{j}$ doctors. Since $h$ received at most $\hat{j}-1$ proposals since the start of phase $L_i$, and since $h$ is a light hospital, the probability $h$ accepts the proposal is at most    
    $$\frac{\ell}{m_h-\hat{j}}\le\frac{\sqrt{n}\log n}{n/2-\hat{j}}.$$

    To bound the expected number of proposals needed in order for the $i$th light-doctor to light hospital match, consider the following experiment: There are two hats, red and blue, each with $m=n/2$ balls. In the red hat there are $r=\sqrt{n}\log n$ red balls and in the blue hat there are $b=\sqrt{n}\log n$ blue balls. We now pull out a ball at random from the red hat until we pick a red ball. Once we pick a red ball, we get to choose a single ball at random from the blue hat. The experiment ends once we pull out a blue ball from the blue hat. The probability of pulling out a red ball on the $j$th try is exactly $$\frac{r}{m-j}=\frac{\sqrt{n}\log n}{n/2-j}.$$
    The probability of pulling out a blue ball on the $\hat{j}$ try is exactly $$\frac{b}{m-\hat{j}}=\frac{\sqrt{n}\log n}{n/2-\hat{j}}.$$
    Therefore, the expected number of balls needed to be pulled out from the red hat is a lower bound for  $\E[L_i|\varepsilon_1,\varepsilon_{2,i},\varepsilon_{3}].$ Let $X_r$ and $X_b$ be random variables taking the value of the number of balls needed to be pulled out of the red and blue hats, respectively, to pull out a red and blue ball, respectively.  The expected number of balls needed to be pulled out of the red hat until we manage to pull out a blue ball out of the blue hat is 
    \begin{eqnarray*}
        \E[X_r\cdot X_b] & = &\E[X_r]\cdot \E[X_b] = \frac{m+1}{r+1}\cdot \frac{m+1}{b+1}\\ & =& \frac{(n/2+1)^2}{(\sqrt{n}\log n+1)^2}= \Omega(n/\log^2 n),
    \end{eqnarray*}
    where the first equality follows by the independence of $X_r$ and $X_b$.
\end{proof}

We now prove that when the probability of $\varepsilon_2$ is high, then we obtain our bound as well.

\begin{lemma} \label{lem:lower_bound_unbalanced_epsilon_2}
    If $\Pr[\varepsilon_2] \ge 1/2$, then $\E[P^A]=\Omega(n^{3/2}/\log n)$.
\end{lemma}
\begin{proof}
    Notice that by Lemma~\ref{lem:eps1_high} and Lemma~\ref{lem:eps3_high}, $$Pr[\varepsilon_1,\varepsilon_{3}]\ge (1-1/2\ln n)^2>1-1/\ln n.$$
    Notice that whenever event $\varepsilon_2$ takes place, event $\varepsilon_{2,i}$ takes place as well. Thus, by our assumption, $\Pr[\varepsilon_{2,i}] \ge 1/2$, and  
    $$Pr[\varepsilon_1, \varepsilon_{2,i}, \varepsilon_{3}]\ge 1/2-1/\ln n.$$

    We conclude:
    \begin{eqnarray*}
        \E[P^A]  & \ge &  \E[\sum_i L_i] \ge \sum_{i=1}^{\ell/8} \E[L_i]\\
        & \ge &\sum_{i=1}^{\ell/8} \E[L_i | \varepsilon_1, \varepsilon_{2,i}, \varepsilon_{3}]\cdot \Pr[\varepsilon_1, \varepsilon_{2,i}, \varepsilon_{3}]\\
        & \ge & \left(\frac{1}{2}-\frac{1}{\ln n}\right)\cdot \frac{\ell}{8}\cdot \frac{(n/2+1)^2}{(\sqrt{n}\log n+1)^2}\\
        & = & \left(\frac{1}{2}-\frac{1}{\ln n}\right)\cdot \frac{\sqrt{n}\log n}{8}\cdot \frac{(n/2+1)^2}{(\sqrt{n}\log n+1)^2}\\
        & = & \Omega(n^{3/2}/\log n),
    \end{eqnarray*}
    where the 4th inequality follows Lemma~\ref{lem:li_bound}. Thus, our bound is obtained.
\end{proof}

Theorem~\ref{thm:lower_bound_unbalanced} follows Lemma~\ref{lem:lower_bound_unbalanced_not_epsilon_2} and Lemma~\ref{lem:lower_bound_unbalanced_epsilon_2}.

\section{Simulations} \label{sec:simulations}

We perform various simulations to empirically demonstrate our theoretical findings. The DA variant used in these simulations operates as follows: there is a FIFO queue of doctors. When a doctor proposes and is rejected, they continue proposing until they are accepted. If a doctor is matched but later rejected, they are placed at the end of the queue.

In Figure~\ref{fig:doc_avg} and Figure~\ref{fig:total_both_phases}, we give simulation results for a market of $1001$ doctors and $1000$ hospitals using the DA algorithm for varying levels of loyalty. Figure ~\ref{fig:doc_avg} present the average rank of doctors under different loyalty levels. Figure~\ref{fig:total_both_phases} compares the total number of proposals in the Balanced Phase and the Unbalanced Phase. The results indicate that the total number of proposals in the Balanced Phase remains in static across different loyalty levels, whereas the Unbalanced Phase accounts for the increase in the number of proposals as loyalty decreases, as our theoretical findings suggest. The simulation results are averaged across 100 experiments, where in each one we use a different random seed to generate random preferences. 

\begin{figure}[H] 
	\centering
	\begin{subfigure}{\columnwidth}	
		\centering
		\includegraphics[width=\columnwidth]{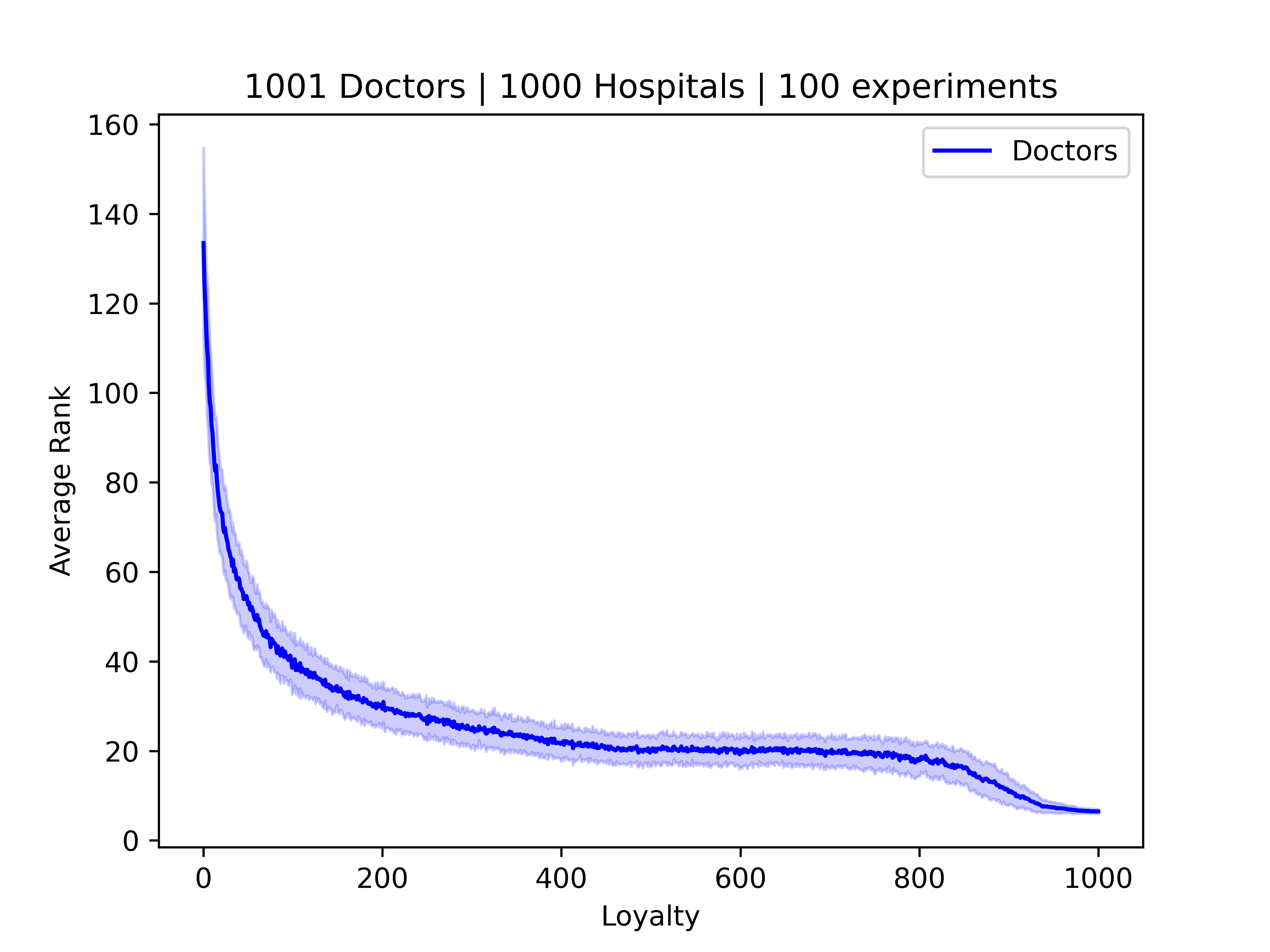}
  \vspace{-2ex}

	\end{subfigure}
    
        \vspace{-2ex}
	\caption{Doctors' average rank for varying amounts of loyalty.}\label{fig:doc_avg}
  \vspace{-1ex}
\end{figure}

\begin{figure}[H]
	\centering
	\begin{subfigure}{\columnwidth}	
		\centering
		\includegraphics[width=\columnwidth]{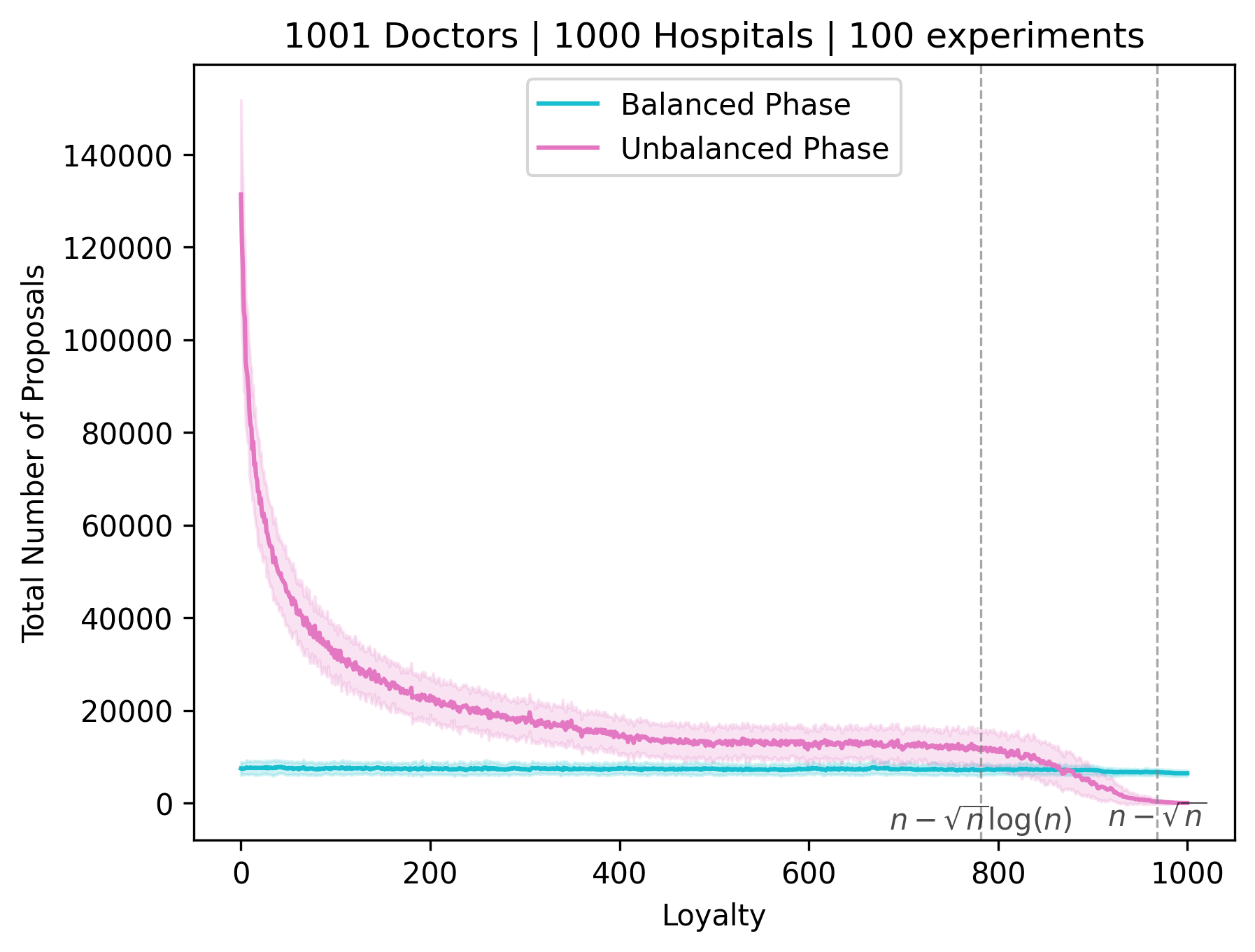}
  \vspace{-2ex}

	\end{subfigure}
    
        \vspace{-2ex}
	\caption{Total number of proposals across the Balanced Phase and Unbalanced Phase for varying amounts of loyalty.}\label{fig:total_both_phases}
  \vspace{-1ex}
\end{figure}

To demonstrate the underlying principals behind the theoretical results, we show hospital rank distributions in a market with $501$ doctors and $500$ hospitals in a single run under varying loyalty levels. The simulations depict the ranks of hospitals for their assigned doctors at the end of the Balanced Phase and Unbalanced Phase for different levels of loyalty. Darker blue shades indicate a higher number of hospitals matched at a given rank. We observe the following:
\begin{itemize}
    \item Figure~\ref{fig:0axes} depicts the case of $k=0$, which is the setting studied in~\cite{ashlagi2017unbalanced}. The long rejection chains can be observed by the many changes in hospitals' ranks, which implies that many of them changed their partner after the balanced phase.
    \item Figure~\ref{fig:250axes} and Figure~\ref{fig:361axes} shows two loyalty parameters that are covered by Corollary~\ref{cor:upper_bound_low_loyalty} and Theorem~\ref{thm:lower_bound_unbalanced}. It is visible that most hospitals of set $T$ after the balanced phase (matched to doctors they rank in the range $[k+\sqrt{n}\cdot \ln n/2+1,k+\sqrt{n}\cdot \ln n]$) are re-matched during the balanced phase, which is exactly what drives our analysis.
    \item Figure~\ref{fig:478axes} show the rank distribution for $k=n-\sqrt{n}$. For this loyalty level, we show an average doctors' rank of $\Theta(\log n)$ in Theorem~\ref{thm:very_high_loyalty_unbalanced}. As shown by our analysis, the hospitals available after the balanced phase (matched to doctors they rank in range $[k+1,n]$) are mostly not re-matched during the unbalanced phase, and this phase ends quickly.
\end{itemize}


\begin{figure}[H]
	\centering
	\begin{subfigure}{0.9\columnwidth}	
		\centering
		\includegraphics[width=\columnwidth]{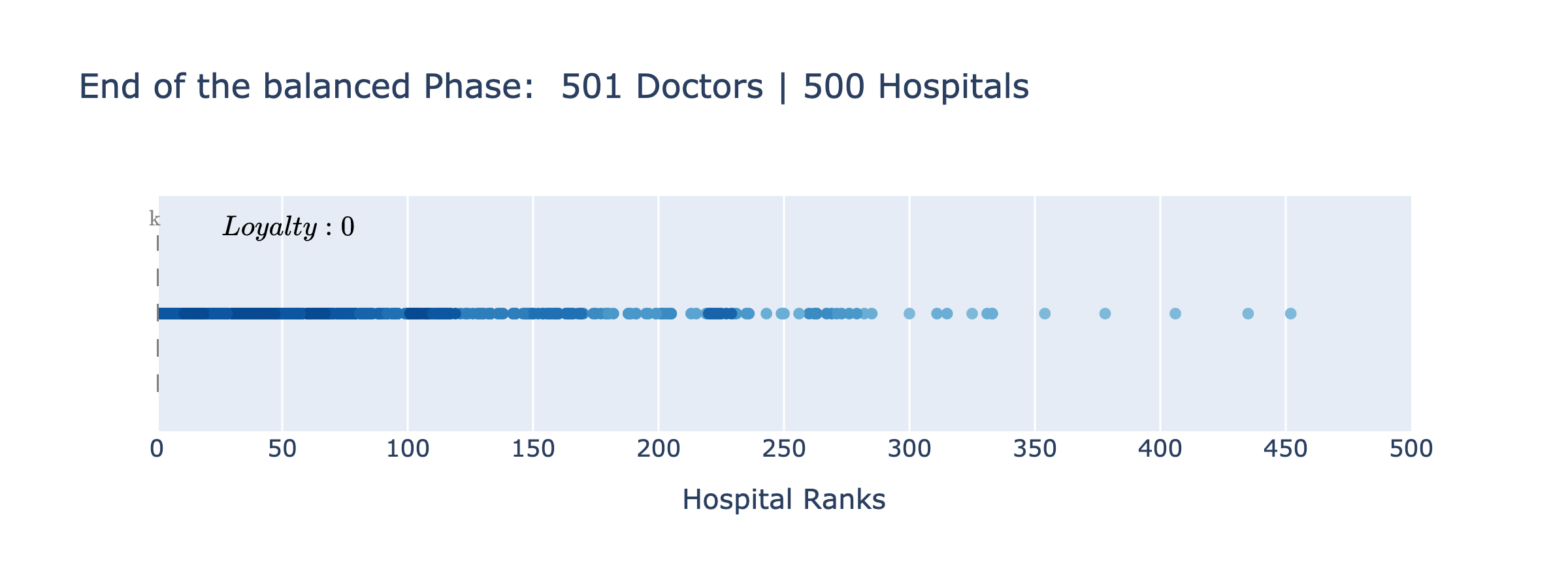}
  \vspace{-4ex}
		
	\end{subfigure}
    \begin{subfigure}{0.9\columnwidth}	
		\centering
		\includegraphics[width=\columnwidth]{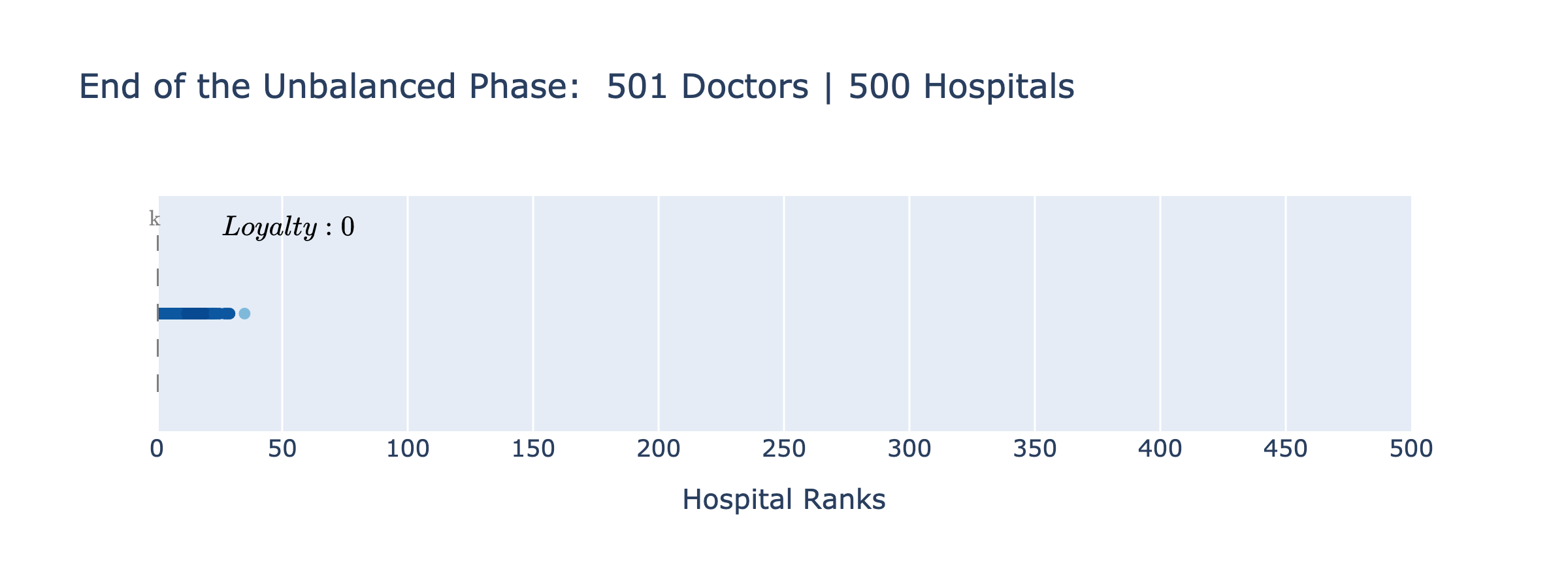}
    \vspace{-4ex}
		
	\end{subfigure}
    
        \vspace{-2ex}
	\caption{Hospital rank distribution for $k=0$.} \label{fig:0axes}
  \vspace{-2ex}
\end{figure}

\begin{figure}[H]
	\centering
	\begin{subfigure}{0.9\columnwidth}	
		\centering
		\includegraphics[width=\columnwidth]{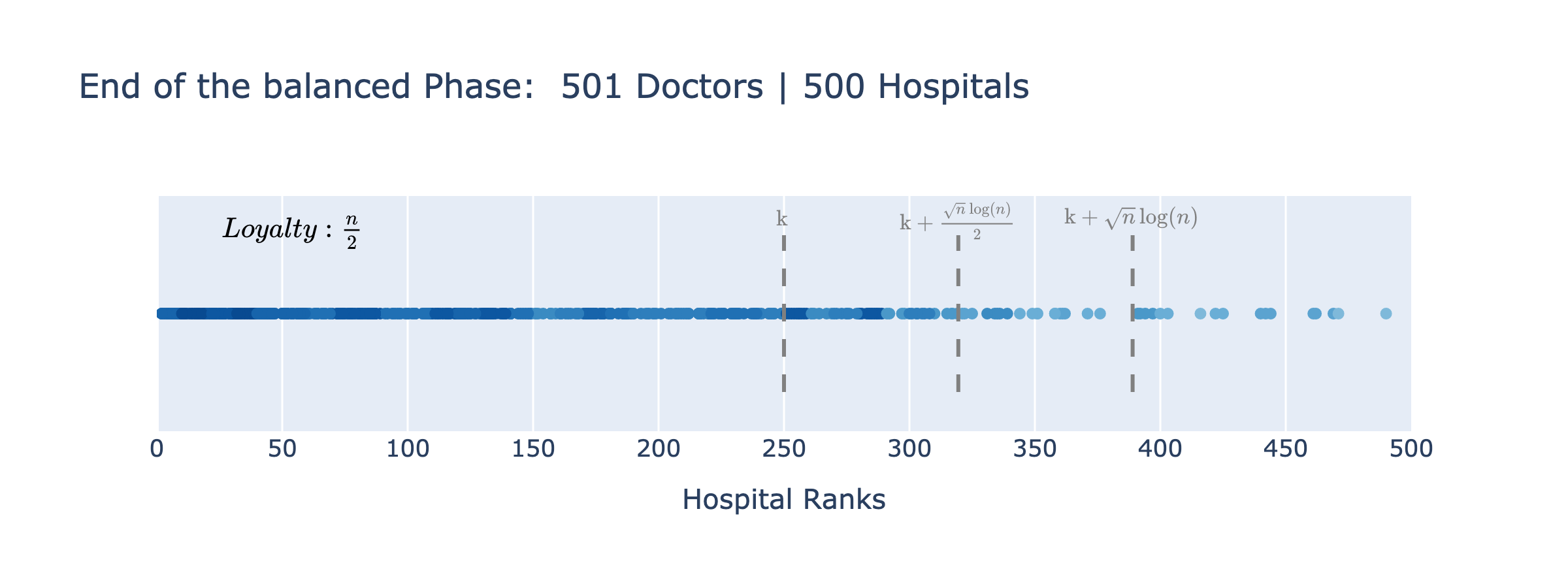}
  \vspace{-4ex}
		
	\end{subfigure}
    \begin{subfigure}{0.9\columnwidth}	
		\centering
		\includegraphics[width=\columnwidth]{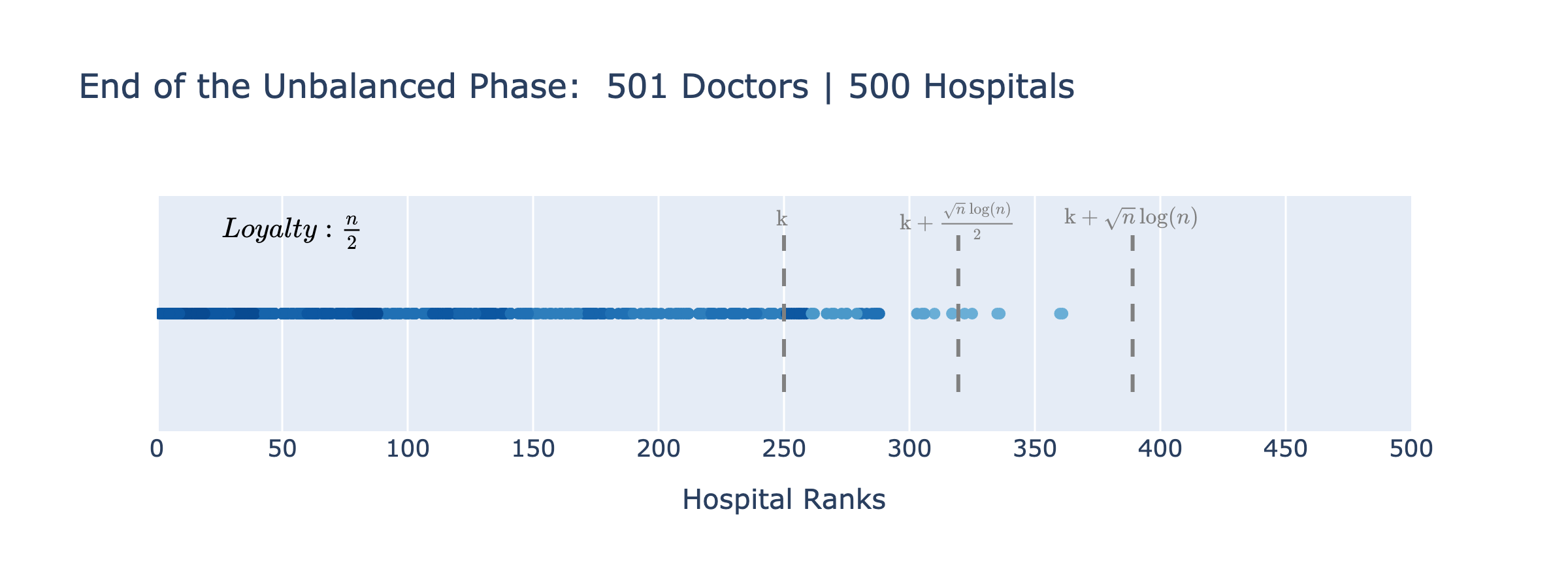}
    \vspace{-4ex}
		
	\end{subfigure}
    
        \vspace{-2ex}
	\caption{Hospital rank distribution for $k=n/2$.}\label{fig:250axes}
  \vspace{-2ex}
\end{figure}

\begin{figure}[H]
	\centering
	\begin{subfigure}{0.9\columnwidth}	
		\centering
		\includegraphics[width=\columnwidth]{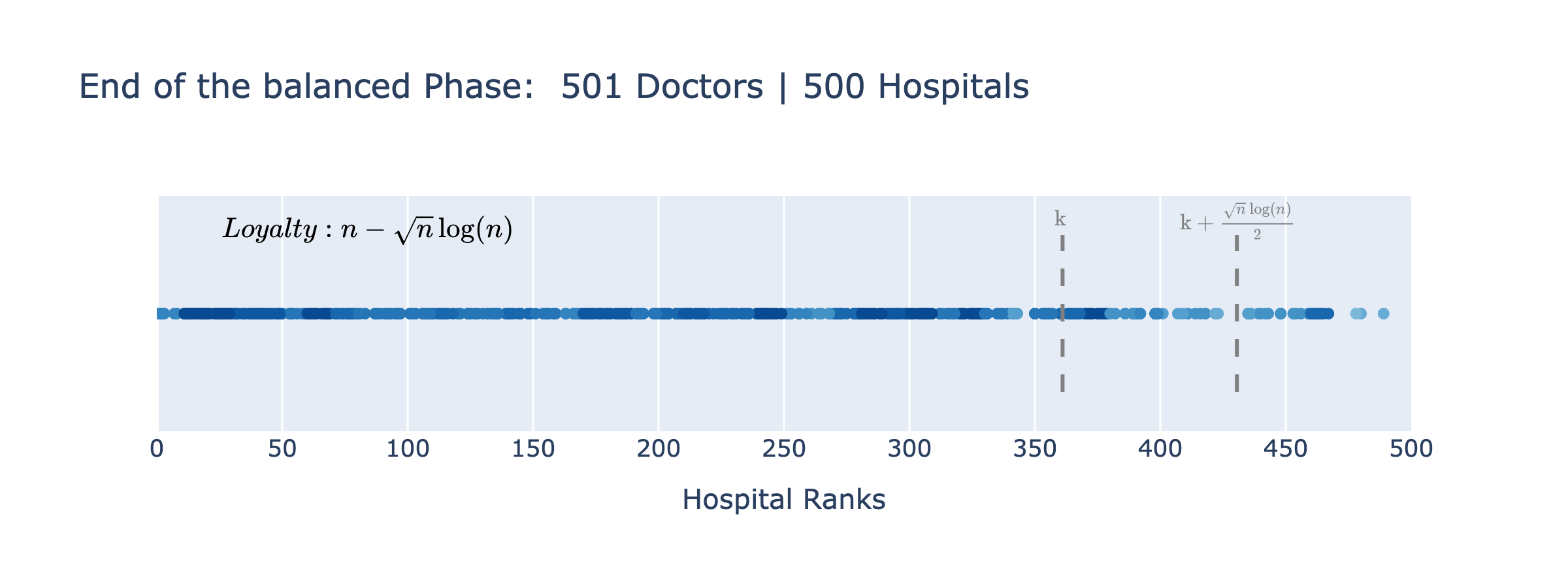}
  \vspace{-4ex}
		
	\end{subfigure}
    \begin{subfigure}{0.9\columnwidth}	
		\centering
		\includegraphics[width=\columnwidth]{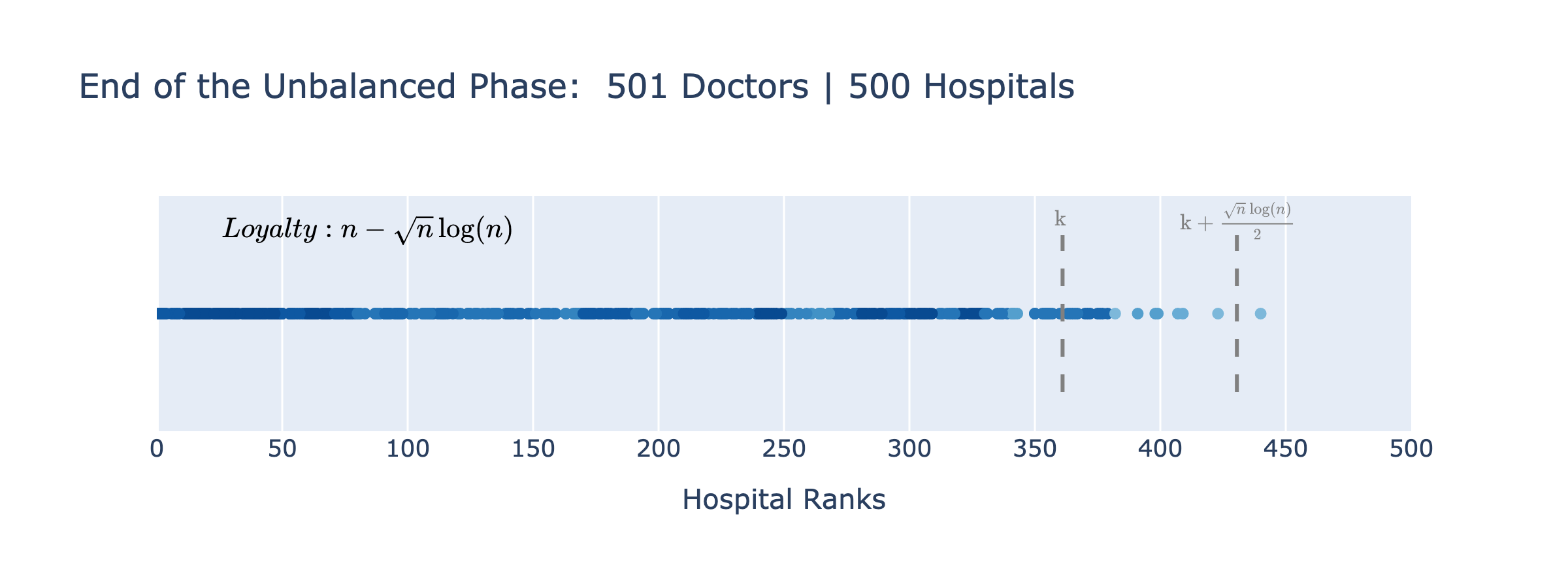}
    \vspace{-4ex}
		
	\end{subfigure}
    
        \vspace{-2ex}
	\caption{Hospital rank distribution for $k=n-\sqrt{n}\log n$.}\label{fig:361axes}
  \vspace{-2ex}
\end{figure}

\begin{figure}[H]
	\centering
	\begin{subfigure}{0.9\columnwidth}	
		\centering
		\includegraphics[width=\columnwidth]{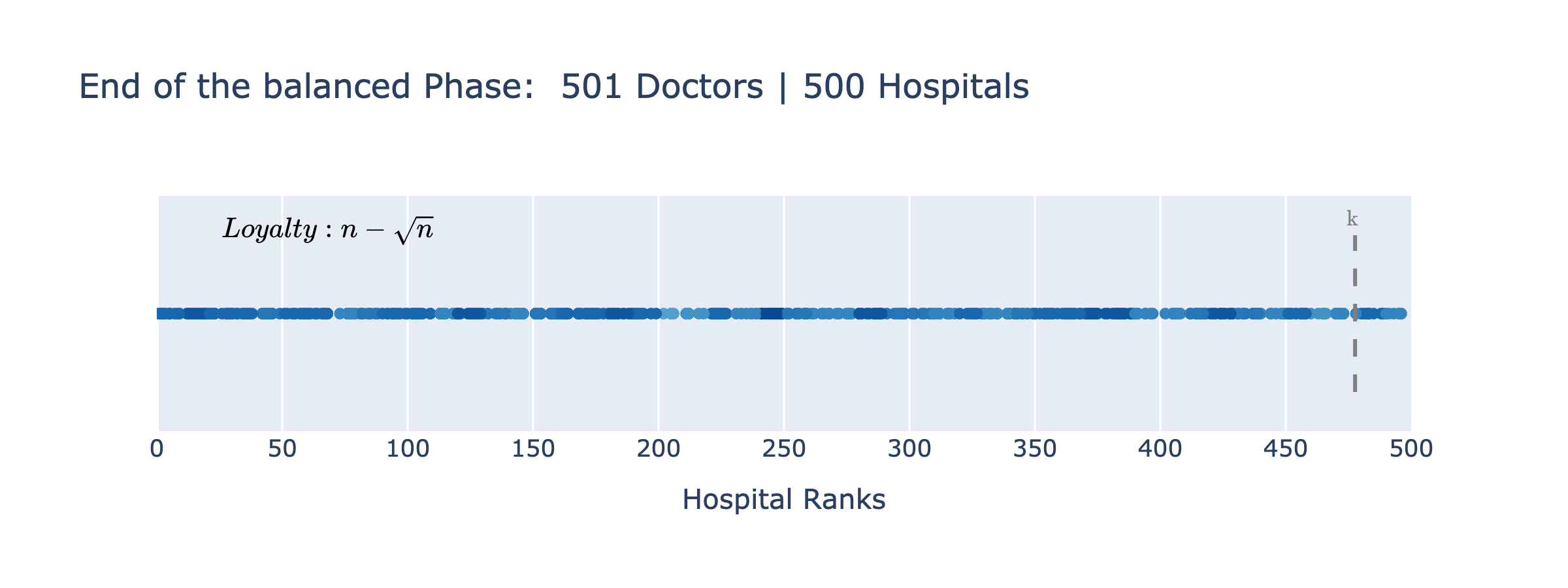}
  \vspace{-4ex}
		
	\end{subfigure}
    \begin{subfigure}{0.9\columnwidth}	
		\centering
		\includegraphics[width=\columnwidth]{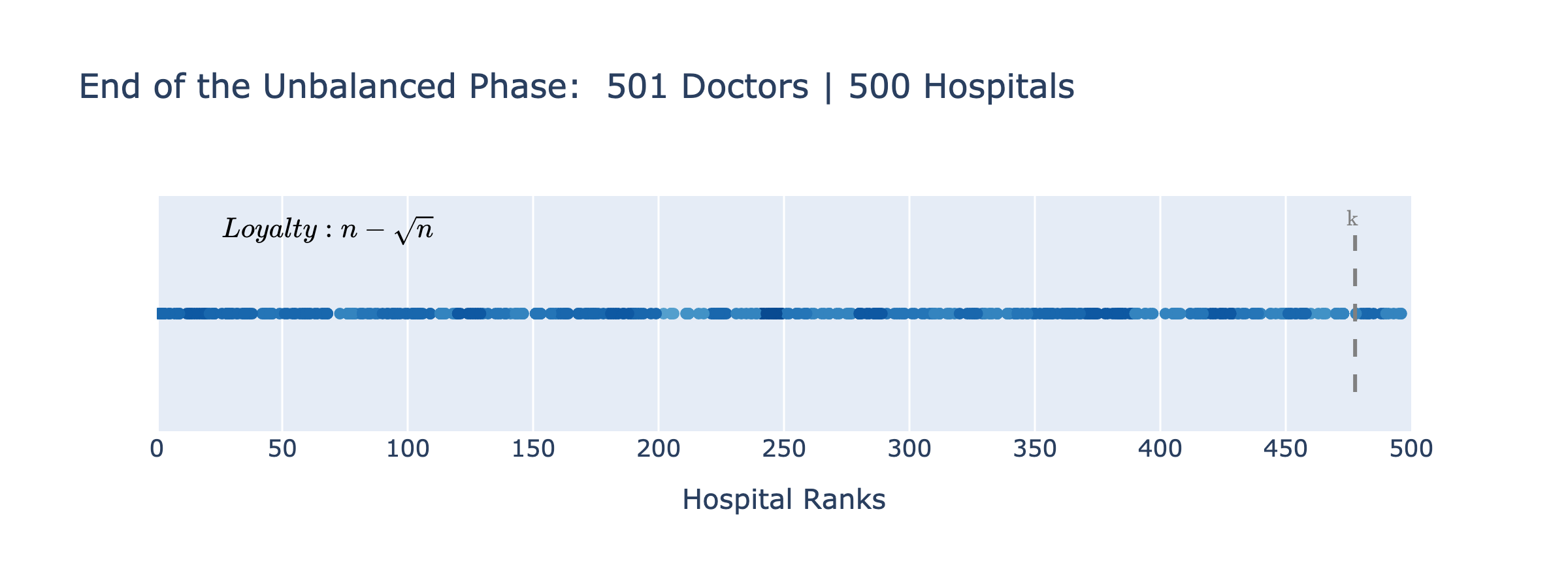}
    \vspace{-4ex}
		
	\end{subfigure}
    
        \vspace{-2ex}
	\caption{Hospital rank distribution for $k=n-\sqrt{n}$.}\label{fig:478axes}
  \vspace{-2ex}
\end{figure}

\section{Discussion}~\label{sec:discussion}
Our findings highlight how loyalty—modeled here as an additive parameter $k$—can partially mitigate the stark “role reversal” observed in unbalanced stable-matching markets. Nonetheless, high levels of loyalty are needed to restore outcomes resembling the balanced case. The immediate questions left open by our research are  understanding exactly how loyalty changes the expected doctors' rank at low levels of loyalty ($k\le n/\log n$), and at which pace does the expected rank changes during the phase transition (for $k\in (n-\sqrt{n}\log n, n-\sqrt{n})$).

Future research could explore replacing the constant $k$ with a \emph{distribution-based} loyalty parameter, such as a Gaussian function based on each hospital’s rank, where hospitals exhibit greater loyalty to mid-ranked doctors and less to extremes. Additionally, studying \emph{truncated preference lists}, where agents rank only a subset of partners, could provide a more realistic model for large-scale platforms and deepen our understanding of stable matching dynamics.

\newpage
\bibliographystyle{plainnat}
\bibliography{loyalty.bib}

\newpage

\appendix

\newpage

\section{Proof of Observation~\ref{obs:balanced_amnesiac_upperbound}} \label{sec:balanced_proofs}

    Consider the amnesiac counterpart $\tilde{A}$ of some consistent DA variant $A$.  For every $i \in [n]$, $P_i^{\tilde{A}}$ follows a geometric distribution with $p=\frac{n-i+1}{n}$ (the probability a proposal is directed to an unmatched hospital); therefore, $\E[P_i^A]=\frac{n}{n-i+1}$. The amount of proposals proposed during the all algorithm is
    \begin{eqnarray*}
    \E[P^A] & =& \E[\sum_i P_i^A]= \sum_{i}\E[P_i^A] \\
    & =&\sum_{i=1}^n \frac{n}{n-i+1}= \sum_{i=1}^n \frac{n}{i}\\ 
        &= &n\cdot H_n = O(n \log n).
    \end{eqnarray*}\qed

\end{document}